\documentclass[11pt, a4paper]{article}
\usepackage{amsmath}
\usepackage{amsthm}
\usepackage{amssymb}
\usepackage{amsfonts}
\usepackage{hyperref}
\usepackage{authblk}
\usepackage{geometry}
\newtheorem{theorem}{Theorem}
\newtheorem{lemma}[theorem]{Lemma}
\newtheorem{corollary}[theorem]{Corollary}

\newtheorem{proposition}[theorem]{Proposition}
\newtheorem{definition}[theorem]{Definition}

\newtheorem{example}[theorem]{Example}
\newtheorem{remark}[theorem]{Remark}

\usepackage{bm}
\usepackage{amssymb, amsmath}
\usepackage{color}
\usepackage{tikz}
\usepackage[ruled, vlined, linesnumbered]{algorithm2e}

\newcommand{\grobner}{Gr\"{o}bner }
\newcommand{\pset}[1]{\mathcal{#1}}
\newcommand{\field}[1]{\mathbb{#1}}
\newcommand{\fk}{\field{K}}
\newcommand{\point}[1]{\bm{#1}}
\newcommand{\p}[1]{\point{#1}}
\newcommand{\kxring}{\field{K}[\point{x}]}
\newcommand{\kx}{\kxring}
\newcommand{\qnum}{\mathbb{Q}}
\newcommand{\bases}[1]{\langle #1 \rangle}

\DeclareMathOperator{\supp}{supp}
\DeclareMathOperator{\lv}{lv}
\DeclareMathOperator{\ldeg}{ldeg}
\DeclareMathOperator{\ini}{ini}
\DeclareMathOperator{\tail}{tail}
\DeclareMathOperator{\zero}{\mathsf{Z}}
\DeclareMathOperator{\level}{level}
\DeclareMathOperator{\prem}{prem}
\DeclareMathOperator{\pquo}{pquo}
\DeclareMathOperator{\red}{red}
\DeclareMathOperator{\sat}{sat}

\DeclareMathOperator{\pop}{\sf pop}
\DeclareMathOperator{\algwang}{TriDecWang}
\DeclareMathOperator{\algReg}{RegDecSubres}
\DeclareMathOperator{\algSub}{TriDecSubres}
\DeclareMathOperator{\srs}{SubRegSubchain}
\DeclareMathOperator{\algSparse}{SparseRegDec}

\begin{document}

\title{Chordal Graphs in Triangular Decomposition\\
  in Top-Down Style\footnote{This work was partially supported by the National Natural Science Foundation of China (NSFC 11401018 and 11771034)}}

\date{}

\author{Chenqi Mou, Yang Bai, and Jiahua Lai}
\affil{LMIB -- School of Mathematics and Systems Science / \authorcr
Beijing Advanced Innovation Center for Big Data and Brain Computing \authorcr
Beihang University, Beijing 100191, China \authorcr \vspace{2mm}
\{chenqi.mou, yangbai, jiahualai\}@buaa.edu.cn}


\maketitle

\begin{abstract}
  In this paper, we first prove that when the associated graph of a polynomial set is chordal, a particular triangular set computed by a general algorithm in top-down style for computing the triangular decomposition of this polynomial set has an associated graph as a subgraph of this chordal graph. Then for Wang's method and a subresultant-based algorithm for triangular decomposition in top-down style and for a subresultant-based algorithm for regular decomposition in top-down style, we prove that all the polynomial sets appearing in the process of triangular decomposition with any of these algorithms have associated graphs as subgraphs of this chordal graph. These theoretical results can be viewed as non-trivial polynomial generalization of existing ones for sparse Gaussian elimination, inspired by which we further propose an algorithm for sparse triangular decomposition in top-down style by making use of the chordal structure of the polynomial set. The effectiveness of the proposed algorithm for triangular decomposition, when the polynomial set is chordal and sparse with respect to the variables, is demonstrated by preliminary experimental results.
\end{abstract}

  \noindent{\small {\bf Key words: }
Triangular decomposition, chordal graph, top-down style, regular decomposition, sparsity}

\section{Introduction}
\label{sec:intro}

In this paper we establish some underlying connections between graph theory and symbolic computation by studying the changes of associated graphs of polynomial sets in the process of decomposing an arbitrary polynomial set with a chordal associated graph into triangular sets with algorithms in top-down style. The study in this paper is directly inspired by the pioneering work of Cifuentes and Parrilo. In \cite{C2017c} they showed for the first time the connections between chordal graphs and triangular sets when they introduced the concept of chordal networks of polynomial sets and proposed an algorithm for constructing chordal networks based on computation of triangular decomposition. In particular, they found experimentally that for polynomial sets with chordal associated graphs, the algorithms for triangular decomposition due to Wang (e.g., his algorithm for regular decomposition in \cite{w00c}) become more efficient. In this paper, with clarification of the changes of associated graphs of polynomial sets in triangular decomposition in top-down style, we are able to provide a theoretical explanation for their experimental observation (see Remark~\ref{rem:explain}). It is worth mentioning that Cifuentes and Parrilo also studied the connections between chordal graphs and \grobner bases in \cite{C2016e}, but they found that the chordal structures of polynomial sets are destroyed in the process of computing \grobner bases . 

Chordal graphs have been applied to many scientific and engineering problems like existence of perfect phylogeny in reconstruction of evolutionary trees \cite{B74a}. Two of these applications are of particular interest to us and are closely related to the study in this paper: sparse Gaussian elimination and sparse sums-of-squares decomposition. For the former problem, it is shown that the Cholesky factorization of a symmetric positive definite matrix does not introduce new fill-ins if the associated graph of the matrix is chordal, and on the basis of this observation algorithms for sparse Gaussian elimination have been proposed by using the property that the sparsity of the matrix can be kept if the associated graph of the matrix is chordal \cite{P1961t,R1970t,G1994p}. For the latter, structured sparsity arising from polynomial optimization problems is studied and utilized by using the chordal structures, resulting in sparse algorithms for sums-of-squares decomposition of multivariate polynomials \cite{WKKM06s,WM10a,ZFP18s,WLX18e}. 

The underlying ideas of the study in this paper are similar to those in the two successful applications of chordal graphs above: we show that the chordality of associated graphs of polynomial sets is preserved in a few algorithms for triangular decomposition in top-down style, as it is in the Cholesky factorization of symmetric matrices, and we propose a sparse algorithm for triangular decomposition in top-down style based on the chordal structure in a simiar way to what have been done for sparse Gaussian elimination and sparse sums-of-squares decomposition. 

Like the \grobner basis which has been greatly developed in its theory, methods, implementations, and applications \cite{B1965A,F1999A,f03a,FM17,CLO1998U}, the triangular set is another powerful algebraic tool in the study on and computation of polynomials symbolically, especially for elimination theory and polynomial system solving \cite{w86z,GC1992s,k93g,w93e,a99t,W2001E,CM2012a}, with diverse applications \cite{Wu94m,c08c}. The process of decomposing a polynomial set into finitely many triangular sets or systems (probably with additional properties like being regular or normal, etc.) with associated zero and ideal relationships is called triangular decomposition of the polynomial set. Triangular decomposition of polynomial sets can be regarded as polynomial generalization of Gaussian elimination for solving linear equations. 

The top-down strategy in triangular decomposition means that the variables appearing in the input polynomial set are handled in a strictly decreasing order, and it is a common strategy in the design and implementations of algorithms for triangular decomposition. In particular, most algorithms for triangular decomposition due to Wang are in top-down style \cite{w93e,w98d,w00c}. Algorithms for triangular decomposition in top-down style with refinement in the Boolean settings and over finite fields have also been proposed and applied to cryptoanalysis \cite{c08c,g12c,h2011a}. The fact that elimination in it is performed in a strictly decreasing order makes triangular decomposition in top-down style the closest among all kinds of triangular decomposition to Gaussian elimination, in which the elimination of entries in different columns of the matrix is also performed in a strict order. 

In this paper the chordal structures of polynomial sets appearing in the algorithms for triangular decomposition in top-down style are studied. The main contributions of this paper include: 1) Under the conditions that the input polynomial set is chordal and a perfect elimination ordering is used as the variable ordering, we study the influence of general reduction in triangular decomposition in top-down style on the associated graphs of polynomial sets and prove that one particular triangular set computed by algorithms for triangular decomposition in top-down style has an associated graph as a subgraph of the input chordal graph (in Section~\ref{sec:general}). 2) Under the same conditions, we show (in Section~\ref{sec:wang}) that in the process of triangular decomposition with Wang's algorithm, any polynomial set (and thus any of the computed triangular sets) has an associated graph as a subgraph of the input chordal graph. 3) The same results are proved for subresultant-based algorithms for triangular decomposition and regular decomposition in top-down style (in Sections~\ref{sec:subres} and \ref{sec:regular} respectively). 4) The variable sparsity of polynomial sets is defined with their associated graphs, and an effective refinement by using the variable sparsity and chordality of input polynomial sets is proposed to speedup triangular decomposition in top-down style (in Section~\ref{sec:sparse}). This paper is an extension of \cite{MB18o}, and the contributions 3) and 4) listed above are new.

With triangular decomposition in top-down style viewed as polynomial generalization of Gaussian elimination, the contributions listed above are indeed polynomial generalizations of the roles chordal structures play in Gaussian elimination and of algorithms for sparse Gaussian elimination. As one may expect, these polynomial generalizations are highly non-trivial because of the complicated process of triangular decomposition due to various splitting strategies involved in specific algorithms. Furthermore, these contributions reveal theoretical properties of triangular decomposition in top-down style from the view point of graph theory, and we hope this paper can stimulate more study on triangular decomposition by using concepts and methods from graph theory.

\section{Preliminaries}
\label{sec:pre}

Let $\fk$ be a field, and $\fk[x_1, \ldots, x_n]$ be the multivariate polynomial ring over $\fk$ in the variables $x_1, \ldots, x_n$. For the sake of simplicity, we write $(x_1, \ldots, x_n)$ as $\p{x}$, $(x_1, \ldots, x_i)$ as $\p{x}_i$ for some integer $i~(1\leq i < n)$, and $\fk[x_1, \ldots, x_n]$ as $\kx$. 

\subsection{Associated graph and chordal graph}
\label{sec:graph}

For a polynomial $F\in \kx$, define the (variable) \emph{support} of $F$, denoted by $\supp(F)$, to be the set of variables in $\{x_1, \ldots, x_n\}$ which effectively appear in $F$. For a polynomial set $\pset{F}\subset \kx$, its support $\supp(\pset{F}) := \bigcup_{F\in \pset{F}}\supp(F)$. 

\begin{definition}\rm
  Let $\pset{F}$ be a polynomial set in $\kx$. Then the \emph{associated graph} of $\pset{F}$, denoted by $G(\pset{F})$, is an undirected graph $(V, E)$ with the vertex set $V = \supp(\pset{F})$ and the edge set $E = \{(x_i, x_j):  1\leq i\neq j\leq n \mbox{ and } \exists F\in \pset{F} \mbox{ such that }x_i, x_j \in \supp(F)\}$.
\end{definition}

\begin{example}\label{ex:associated} \rm
  The associated graphs of 
  \begin{equation*}
    \begin{split}
      \pset{P} &= \{x_2+x_1, x_3+x_1, x_4^2+x_2, x_4^3+x_3, x_5+x_2, x_5+x_3+x_2 \},\\
\pset{Q} &= \{x_2+x_1, x_3+x_1, x_3, x_4^2+x_2, x_4^3+x_3, x_5+x_2\}
    \end{split}
  \end{equation*}
are shown in Figure~\ref{fig:associated}.
    \begin{figure}[ht]
      \centering
\includegraphics[width=3cm,keepaspectratio]{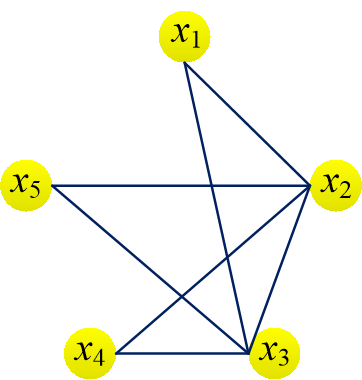}~\qquad
\includegraphics[width=3cm,keepaspectratio]{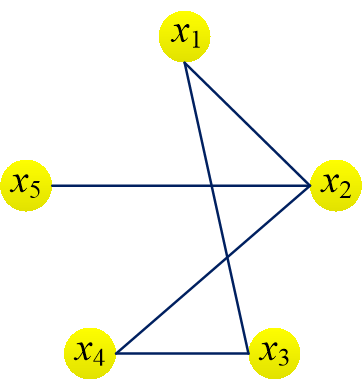}      
      \caption{The associated graphs $G(\pset{P})$ (left) and $G(\pset{Q})$ (right) in Example~\ref{ex:associated}}
      \label{fig:associated}
    \end{figure}
\end{example}

\begin{definition}\label{def:poe}\rm
Let $G = (V, E)$ be a graph with $V = \{x_1, \ldots, x_n\}$. Then an ordering $x_{i_1} < x_{i_2} < \cdots < x_{i_n}$ of the vertices is called a \emph{perfect elimination ordering} of $G$ if for each $j=i_1, \ldots, i_n$, the restriction of $G$ on the following set
\begin{equation}\label{eq:smaller}
  X_j = \{x_j\} \cup \{x_k: x_k < x_j \mbox{ and } (x_k, x_j) \in E\}
\end{equation}
is a clique. A graph $G$ is said to be \emph{chordal} if there exists a perfect elimination ordering of it. 
\end{definition}

An equivalent condition for a graph $G = (V, E)$ to be \emph{chordal} is the following: for any cycle $C$ contained in $G$ of four or more vertices, there is an edge $e\in E \setminus C$ connecting two vertices in $C$. The edge $e$ in this case is called a \emph{chord} of $C$. A chordal graph is also called a triangulated one. For an arbitrary graph $G$, another graph $G'$ is called a \emph{chordal completion} of $G$ if $G'$ is chordal and $G$ is its subgraph. 

From the algorithmic point of view, there exist effective algorithms for testing whether an arbitrary graph is chordal (in case of a chordal graph, a perfect elimination ordering will also be returned) \cite{RTL76a} and for finding a chordal completion of an arbitrary graph \cite{BK08c}, though the problem of finding the minimal chordal completion is NP-hard \cite{ACP87c}. 

\begin{definition}\rm
  A polynomial set $\pset{F} \subset \kx$ is said to be \emph{chordal} if its associated graph $G(\pset{F})$ is chordal. 
\end{definition}

\begin{example}\rm
In Example~\ref{ex:associated} and Figure~\ref{fig:associated}, the associated graph $G(\pset{P})$ is chordal by definition and thus $\pset{P}$ is chordal, while $G(\pset{Q})$ is not. 
\end{example}

\subsection{Triangular set and triangular decomposition}
\label{sec:ts}
Throughout this subsection the variables are ordered as $x_1 < \cdots < x_n$. For an arbitrary polynomial $F\in \kx$, the greatest variable appearing in $F$ is called its \emph{leading variable}, denoted by $\lv(F)$. Let $\lv(F) = x_k$. Write $F = Ix_k^d + R$ with $I\!\in\! \fk[x_1, \ldots, x_{k-1}]$, $R\!\in\! \fk[x_1, \ldots, x_k]$, and $\deg(R, x_k) < d$. Then the polynomials $I$ and $R$ are called the \emph{initial} and \emph{tail} of $F$ and denoted by $\ini(F)$ and $\tail(F)$ respectively, and $d$ is called the \emph{leading degree} of $F$ and denoted by $\ldeg(F)$. For two polynomial sets $\pset{F}, \pset{G}\subset \kx$, the set of common zeros of $\pset{F}$ in $\overline{\fk}^n$ is denoted by $\zero(\pset{F})$, and $\zero(\pset{F} / \pset{G}) := \zero(\pset{F}) \setminus \zero(\prod_{G\in \pset{G}}G)$, where $\overline{\fk}$ is the algebraic closure of $\fk$.

\begin{definition}\label{def:Tset-system}\rm
  An ordered set of non-constant polynomials $\pset{T} = [T_1, \ldots, T_r] \subset \kx$ is called a \emph{triangular set} if $\lv(T_1) < \cdots < \lv(T_r)$.  A pair $(\pset{T}, \pset{U})$ with $\pset{T}, \pset{U} \subset \kx$ is called a \emph{triangular system} if $\pset{T}$ is a triangular set, and for each $i = 2, \ldots, r$ and any $\overline{\p{x}}_{i-1} \in \zero([T_1, \ldots, T_{i-1}] / \pset{U})$, we have $\ini(T_i)(\overline{\p{x}}_{i-1}) \neq 0$. 
\end{definition}

Given a triangular set $\pset{T} = [T_1, \ldots, T_r] \subset \kx$, the \emph{saturated ideal} of $\pset{T}$ is $\sat(\pset{T}) := \bases{\pset{T}}: (\prod_{i=1}^r T_i)$. In particular, for an integer $i~(1\leq i < r)$, $[T_1, \ldots, T_i]$ forms a (truncated) triangular set in $\fk[x_1, \ldots, \lv(T_i)]$, and we denote $\sat_i(\pset{T}) := \sat([T_1, \ldots, T_i])$. For an arbitrary polynomial set $\pset{P} \subset \kx$, we denote $\pset{P}^{(i)} := \{P\in \pset{P}:\, \lv(P) = x_i\}$ for an integer $i~(1\leq i \leq n)$ and denote $\pset{P}^{(0)} := \{P \in \pset{P}: P \in \fk \}$. 

\begin{definition}\rm
  \label{def:reg}
  A triangular set $\pset{T} = [T_1, \ldots, T_r] \subset \kx$ is said to be \emph{regular} or called a \emph{regular set} if for each $i = 2, \ldots, r$, the canonical image of $\ini(T_i)$ in $\fk[x_1, \ldots, \lv(T_{i-1})] / \sat_{i-1}(\pset{T})$ is neither zero nor a zero-divisor. A triangular system $(\pset{T}, \pset{U})$ is called a \emph{regular system} if for each $i=1, \ldots, n$, the following conditions hold: (a) either $\pset{T}^{(i)} = \emptyset$ or $\pset{U}^{(i)} = \emptyset$; (b) for any $\overline{\p{x}}_{i-1} \in \zero([T_1, \ldots, T_{i-1}] / \bigcup_{j=1}^{i-1}\pset{U}^{(j)})$ and $U \in \pset{U}^{(i)}$, we have $\ini(U)(\overline{\p{x}}_{i-1}) \neq 0$. 
\end{definition}

The definitions above of regular set and regular system are algebraic (in the language of ideals) and geometric (in the language of zeros) respectively. The connections between regular sets and regular systems have been clarified in \cite{w00c,W2001E}.

\begin{definition}\label{def:DintoSet}\rm
 Let $\pset{F}\subset \kx$ be a polynomial set. Then a finite number of triangular sets $\pset{T}_1, \ldots, \pset{T}_s \subset \kx$ (triangular systems $(\pset{T}_1,\pset{U}_1), \ldots, (\pset{T}_s, \pset{U}_s)$ respectively) are called a \emph{triangular decomposition} of $\pset{F}$ if the zero relationship $\zero(\pset{F}) = \bigcup_{i=1}^s \zero(\pset{T}_i / \ini(\pset{T}_i))$ holds, where $\ini(\pset{T}_i) := \{\ini(T): T \in \pset{T}_i\}$ ($\zero(\pset{F}) = \bigcup_{i=1}^s \zero(\pset{T}_i / \pset{U}_i)$ holds respectively). In particular, a triangular decomposition is called a \emph{regular decomposition} if each of its triangular sets or systems is regular. 
\end{definition}

When no ambiguity occurs, the process for computing the triangular decomposition of a polynomial set $\pset{F}$ is also called \emph{triangular decomposition} of $\pset{F}$. As one may find from Definitions~\ref{def:Tset-system} and \ref{def:DintoSet}, triangular systems are generalization of triangular sets. For a triangular system $(\pset{T}, \pset{U})$, $\pset{T}$ is a triangular set which represents the equations $\pset{T}=0$, while $\pset{U}$ is a polynomial set which represents the inequations $\pset{U}\neq 0$.

There exist many algorithms for decomposing polynomial sets into triangular sets or systems with different properties. One of the main strategies for designing such algorithms for triangular decomposition is to carry out reduction on polynomials containing the greatest (unprocessed) variable until there is only one such polynomial left, at the same time producing new polynomials whose leading variables are strictly smaller than the currently processed variable. 

For an arbitrary polynomial set $\pset{P} \in \kx$, the smallest integer $i~(0\leq i \leq n)$ such that $\#\pset{P}^{(j)} = 0$ or $1$ for each $j=i+1, \ldots, n$ is called the \emph{level} of $\pset{P}$ and denoted by $\level(\pset{P})$. Obviously a polynomial set $\pset{P}$ containing no constant forms a triangular set if $\level(\pset{P})=0$. 

Let $\pset{F}$ be a polynomial set in $\kx$ and $\Phi$ be a set of pairs of polynomial sets, initialized with $\{(\pset{F}, \emptyset)\}$. Then an algorithm for computing triangular decomposition of $\pset{F}$ is said to be in \emph{top-down style} if for each polynomial set $(\pset{P}, \pset{Q}) \in \Phi$ with $\level(\pset{P}) = k > 0$, this algorithm handles the polynomials in $\pset{P}^{(k)}$ and $\pset{Q}^{(k)}$ to produce finitely many polynomials sets $\pset{P}_1, \ldots, \pset{P}_s$ and $\pset{Q}_1, \ldots, \pset{Q}_s$ such that the following conditions hold:
\begin{enumerate}
\item[(a)] $\zero(\pset{P} / \pset{Q}) = \bigcup_{i=1}^{s}\zero(\pset{P}_i / \pset{Q}_i)$;
\item[(b)] for each $i=1, \ldots, s$, $\pset{P}_i^{(j)} = \pset{P}^{(j)}$ and $\pset{Q}_i^{(j)} = \pset{Q}^{(j)}$ for $j=k+1, \ldots, n$; 

\item[(c)] there exists some integer $\ell~(1\leq \ell \leq s)$ such that $\#\pset{P}_{\ell}^{(k)}= 0$ or $1$, and the other $(\pset{P}_i, \pset{Q}_i)~(i \neq \ell)$ are put into $\Phi$ for later computation. 
\end{enumerate}

In this paper we are interested mainly in algorithms for triangular decomposition in top-down style. Note that the above definition, compared with the corresponding one in \cite{MB18o}, imposes additional conditions on the polynomial sets $Q_1, \ldots, Q_s$ representing inequations, for the authors find that it is difficult to study the polynomial sets $\pset{P}_1, \ldots, \pset{P}_s$ alone when the interactions between $\pset{P}$ and $\pset{Q}$ occur in certain algorithms (see Section~\ref{sec:regular} for more details).

\subsection{Pseudo division and subresultant regular subchain}
\label{sec:srs}

Two commonly used algebraic operations on multivariate polynomials to perform reduction in algorithms for triangular decomposition are pseudo division and computation of the resultant of two polynomials. The algorithms for triangular decomposition in top-down style studied in this paper rely heavily on these two algebraic operations. 

For any two polynomials $F, G \in \kx$, there exist polynomials $Q, F \in \kx$ and an integer $s~(0\leq s \leq \deg(F, \lv(G))+\ldeg(G)-1)$ such that $\ini(G)^s F = QG + R$ and $\deg(R, \lv(G)) < \ldeg(G)$. Furthermore, if $s$ is fixed, then $Q$ and $R$ are unique. The process above of computing $Q$ and $R$ from $F$ and $G$ is called the \emph{pseudo division} of $F$ with respect to $G$, and the polynomials $Q$ and $R$ here are called the \emph{pseudo quotient} and \emph{pseudo remainder} of $F$ with respect to $G$ and denoted by $\pquo(F, G)$ and $\prem(F, G)$ respectively. 

Suppose further that $x_k \in \supp(F) \cap \supp(G)$. Write
$F=\sum_{i=0}^ma_i\,x_k^i$ and $G=\sum_{j=0}^lb_j\,x_k^j$ with $a_i, b_j \in \fk[x_1, \ldots, x_{k-1}, x_{k+1}, \ldots, x_n]$. Denote by $\mathbf{M}$ the sylvester matrix of $F$ and $G$ with respect to $x_k$. Then the determinant $|\mathbf{M}|$ is called the \emph{Sylvester resultant} of $F$ and $G$ with respect to $x_k$. 

For two integers $i, j~(0\leq i \leq j < l)$, define $\mathbf{M}_{ij}$ to be the submatrix of $\mathbf{M}$ obtained by deleting the last $j$ rows of $F$'s coefficients, the last $j$ rows of $G$'s coefficients, and the last $2j + 1$ columns except the $(m+l-i-j)$-th one. Then the polynomial $S_j = \sum_{i=0}^j |\mathbf{M}_{ij}|x_k^i$ is called the $j$th \emph{subresultant} of $F$ and $G$ with respect to $x_k$. In particular, the $j$th subresultant $S_j$ is said to be \emph{regular} if $\deg(S_j, x_k) = j$. 

\begin{definition}\rm
  Let $F, G \in \kx$ be two polynomials such that $m = \deg(F, x_k) \geq \deg(G, x_k) =l$, and $S_j \in \kx$ be the $j$th resultant of $F$ and $G$ with respect to $x_k$ for $j = 0, \ldots, \mu-1$, where $\mu := m-1$ when $m > l$ and $\mu := l$ otherwise. Then the sequence $F, G, S_{\mu-1}, S_{\mu-2}, \ldots, S_0$ is called the \emph{subresultant chain} of $F$ and $G$ with respect to $x_k$. Furthermore, let $S_{d_1}, \ldots, S_{d_r}$ be the regular subresultants in $S_{\mu-1}, \ldots, S_0$ with $d_1 > \cdots > d_r$. Then the sequence $F, G, S_{d_1}, \ldots, S_{d_r}$ is called the \emph{subresultant regular subchain} of $F$ and $G$ with respect to $x_k$. 
\end{definition}

There exist strong connections between the subresultant chain and the greatest common divisor of two polynomials. The reader is referred to \cite[Chap.~7]{M2012A} for more details on this.

\section{General triangular decomposition in top-down style}
\label{sec:general}

In this section, the graph structures of polynomial sets in general algorithms for triangular decomposition in top-down style are studied when the input polynomial set is chordal. We start this section with the connections between the associated graphs of a triangular set reduced from a chordal polynomial set and the chordal associated graph. 

\begin{proposition}\label{prop:woReduction}
Let $\pset{P} \subset \kx$ be a chordal polynomial set with $x_1 < \cdots < x_n$ as one perfect elimination ordering of $G(\pset{P})$. For $i=1, \ldots, n$, let $T_i\in \kx$ be a polynomial such that $\lv(T_i) = x_i$ and $\supp(T_i) \subset \supp(\pset{P}^{(i)})$ ($T_i$ is set null if $\pset{P}^{(i)} = \emptyset$). Then $\pset{T} = [T_1, \ldots, T_n]$ is a triangular set, and $G(\pset{T}) \subset G(\pset{P})$. In particular, if $\supp(T_i) = \supp(\pset{P}^{(i)})$ for $i=1, \ldots, n$, then $G(\pset{T}) = G(\pset{P})$.
\end{proposition}

\begin{proof}
It is straightforward that $\pset{T}$ is a triangular set because $\lv(T_i) = x_i$ if $\pset{P}^{(i)}\neq \emptyset$ for $i=1, \ldots, n$. 

For any edge $(x_i,x_j) \in G(\pset{T})$, there exists an integer $k~(i, j \leq k \leq n)$ such that $x_i, x_j \in \supp(T_k)$. Then $x_i, x_j \in \supp(\pset{P}^{(k)})$, and thus $(x_i, x_k) \in G(\pset{P})$ and $(x_j, x_k) \in G(\pset{P})$. Since $G(\pset{P})$ is chordal with $x_1 < \ldots < x_n$ as a perfect elimination ordering and $x_i \leq x_k$, $x_j \leq x_k$, we know that $(x_i, x_j) \in G(\pset{P})$ by Definition~\ref{def:poe}. This proves the inclusion $G(\pset{T}) \subset G(\pset{P})$.

In the case when $\supp(T_i) = \supp(\pset{P}^{(i)})$ for $i=1, \ldots, n$, next we show the inclusion $G(\pset{T}) \supset G(\pset{P})$, which implies the equality $G(\pset{T}) = G(\pset{P})$. For any $(x_i, x_j) \in G(\pset{P})$, there exists an integer $k$ and a polynomial $P$ such that $x_i, x_j \in \supp(P)$ with $P \in \pset{P}^{(k)}$. Since $\supp(P) \subseteq \supp(T_k)$, we know that $x_i, x_j \in \supp(T_k)$ and thus $(x_i, x_j) \in G(\pset{T})$. 
\end{proof}

\begin{example}\label{ex:ts}\rm
  Proposition~\ref{prop:woReduction} does not necessarily hold in general if the polynomial set $\pset{P}$ is not chordal. Consider the same $\pset{Q}$ as in Example~\ref{ex:associated} whose associated graph $G(\pset{Q})$ is not chordal. Let 
$$\pset{T} = [x_2+x_1, x_3+x_1, -x_2x_4+x_3, x_5+x_2].$$
Then one can check that for $i = 2, \ldots, 5$, $\supp(\pset{T}^{(i)}) = \supp(\pset{Q}^{(i)})$, but the associated graph $G(\pset{T})$, as shown in Figure~\ref{fig:ts}, is not a subgraph of $G(\pset{Q})$. 
    \begin{figure}[ht]
      \centering
\includegraphics[width=3cm,keepaspectratio]{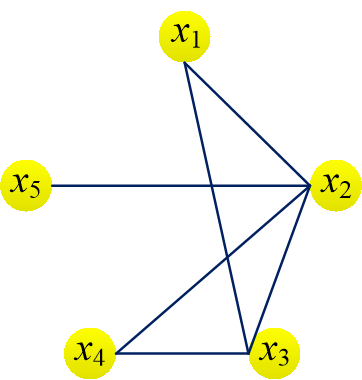}      
      \caption{The associated graph $G(\pset{T})$ in Example~\ref{ex:ts}}
\label{fig:ts}
    \end{figure}
\end{example}

The following theorem relates the associated graph of a chordal polynomial set and that of the polynomial set after reduction with respect to one variable.

\begin{theorem}\label{thm:reduction}
Let $\pset{P} \subset \kx$ be a chordal polynomial set such that $\pset{P}^{(n)} \neq \emptyset$ and $x_1 < \ldots < x_n$ is one perfect elimination ordering of $G(\pset{P})$. Let $T \in \kx$ be a polynomial such that $\lv(T) = x_n$ and $\supp(T) \subset \supp(\pset{P}^{(n)})$, and $\pset{R} \subset \kx$ be a polynomial set such that $\supp(\pset{R}) \subset \supp(\pset{P}^{(n)}) \setminus \{x_n\}$. Then for the polynomial set $\tilde{\pset{P}} = \tilde{\pset{P}}^{(1)} \cup \cdots \cup \tilde{\pset{P}}^{(n-1)} \cup \{T\}$, where $\tilde{\pset{P}}^{(k)} = \pset{P}^{(k)} \cup \pset{R}^{(k)}$ for $k=1, \ldots, n-1$, we have $G(\tilde{\pset{P}}) \subset G(\pset{P})$. In particular, if $\supp(T) = \supp(\pset{P}^{(n)})$, then $G(\tilde{\pset{P}}) = G(\pset{P})$.
\end{theorem}

\begin{proof}
To prove the inclusion $G(\tilde{\pset{P}}) \subset G(\pset{P})$, it suffices to show that for each edge $(x_i, x_j) \in G(\tilde{\pset{P}})$, we have $(x_i, x_j) \in G(\pset{P})$. For an arbitrary edge $(x_i, x_j) \in G(\tilde{\pset{P}})$, there exists a polynomial $P\in \tilde{\pset{P}}$ and an integer $k~(i, j \leq k \leq n)$ such that $x_i, x_j \in \supp(P)$ and $P\in \tilde{\pset{P}}^{(k)}$. 

If $k=n$, then $x_i, x_j \in \supp(T)$, and by $\supp(T) \subset \supp(\pset{P}^{(n)})$ we have $x_i, x_j \in \supp(\pset{P}^{(n)})$. This implies that $(x_i, x_n), (x_j, x_n) \in G(\pset{P}^{(n)}) \subset G(\pset{P})$ and by the chordality of $G(\pset{P})$ we have $(x_i, x_j) \in G(\pset{P})$.

Else if $k<n$, then by $\tilde{\pset{P}}^{(k)} = \pset{P}^{(k)} \cup \pset{R}^{(k)}$ there are two cases for $P$ accordingly: when $P\in \pset{P}^{(k)} \subset \pset{P}$, clearly $(x_i, x_j) \in G(\pset{P})$; when $P\in \pset{R}^{(k)}$, we have $x_i, x_j \in \supp(\pset{R}^{(k)}) \subset \supp(\pset{P}^{(n)})$, and thus $(x_i, x_n), (x_j, x_n) \in G(\pset{P}^{(n)}) \subset G(\pset{P})$, and the chordality $G(\pset{P})$ implies $(x_i, x_j) \in G(\pset{P})$. 

In particular, if $\supp(T) = \supp(\pset{P}^{(n)})$, then by $G(\pset{P}^{(k)}) \subset G(\tilde{\pset{P}}^{(k)})$ for $k=1, \ldots, n-1$ and $G(\pset{P}^{(n)})\subset G(T)$ we have $G(\pset{P}) \subset G(\tilde{\pset{P}})$. This proves the equality $G(\tilde{\pset{P}}) = G(\pset{P})$.
\end{proof}

\begin{example}\label{ex:oneReduction}\rm
  Let $\pset{P}$ be the chordal polynomial set as in Example~\ref{ex:associated}. Then $\pset{P}^{(5)} = \{x_5+x_2, x_5+x_3+x_2\}$. If we take $T = x_5+x_2$, and $\pset{R} = \{\prem(x_5+x_3+x_2, x_5+x_2)\} = \{x_3\}$, then $\tilde{\pset{P}}$ equals $\pset{Q}$ in Example~\ref{ex:associated}, and $G(\tilde{\pset{P}})$ is a (strict) subgraph of $G(\pset{P})$; If we take $T = x_5+x_3+x_2$, and $\pset{R} = \{\prem(x_5+x_2, x_5+x_3+x_2)\} = \{-x_3\}$, then $\supp(T) = \supp(\pset{P}^{(5)})$ and thus $G(\tilde{\pset{P}}) = G(\pset{P})$.
\end{example}

Next we introduce some notations to formulate the reduction process in Theorem~\ref{thm:reduction}. Denote the power set of a set $S$ by $2^S$. For an integer $i~(1\leq i \leq n)$, let $f_i$ be a mapping
 \begin{equation}\label{eq:mapping}
   \begin{split}
 f_i: 2^{\fk[\p{x}_i] \setminus \fk[\p{x}_{i-1}]} &\rightarrow (\fk[\p{x}_i] \setminus \fk[\p{x}_{i-1}]) \times 2^{\fk[\p{x}_{i-1}]}\\
 \pset{P} &\mapsto (T, \pset{R})     
   \end{split}
 \end{equation}
 such that $\supp(T) \subset \supp(\pset{P}) $ and $\supp(\pset{R}) \subset \supp(\pset{P})$, where $\fk[\p{x}_0]$ is understood as $\fk$. For a polynomial set $\pset{P} \subset \kx$ and a fixed integer $i~(1\leq i \leq n)$, suppose that $(T_i, \pset{R}_i) = f_i(\pset{P}^{(i)})$ for some $f_i$ as stated above. Now define the result of reduction with respect to $x_i$ as the polynomial set $\red_i(\pset{P})$ by defining all its subsets $\red_i(\pset{P})^{(j)}$ for $j=1, \ldots, n$ as follows.
\begin{equation}
  \label{eq:reduction}
\red_i(\pset{P})^{(j)} := \left\{
    \begin{tabular}[l]{ll}
      $\pset{P}^{(j)}$, & if $j > i$  \\
      $\{T_i\}$, & if $j=i$ \\
      $\pset{P}^{(j)} \cup \pset{R}_i^{(j)}$, & if $j<i$
    \end{tabular}
\right.
\end{equation}
Furthermore, denote 
\begin{equation}
  \label{eq:redBar}
  \overline{\red}_i(\pset{P}) := \red_{i}(\red_{i+1}(\cdots (\red_n(\pset{P}))\cdots))
\end{equation}
for simplicity, and the polynomial set $\overline{\red}_i(\pset{P})$ is the result of successive reduction with respect to $x_n, x_{n-1}, \ldots, x_i$. Following the above terminologies, the conclusions of Theorem~\ref{thm:reduction} can be reformulated as: $G(\red_n(\pset{P})) \subset G(\pset{P})$, and the equality holds if $\supp(T_n) = \supp(\pset{P}^{(n)})$.


Indeed, the reduction process above is commonly used in algorithms for triangular decomposition in top-down style, and the mapping $f_i$ in \eqref{eq:mapping} is abstraction of specific reductions used in different kinds of algorithms for triangular decomposition \cite{JLW2013a}. For example, one specific kind of such reduction is performed by using pseudo divisions, and in this case $\pset{R}$ in \eqref{eq:mapping} consists of pseudo remainders which do not contain $x_i$.

\begin{proposition}\label{prop:left}
Let $\pset{P} \subset \kx$ be a chordal polynomial set with $x_1 < \cdots < x_n$ as one perfect elimination ordering of $G(\pset{P})$. For each $i~(1\leq i \leq n)$, suppose that $(T_i, \pset{R}_i) = f_i(\overline{\red}_{i+1}(\pset{P})^{(i)})$ for some $f_i$ as in \eqref{eq:mapping} and $\supp(T_i) = \supp(\overline{\red}_{i+1}(\pset{P})^{(i)})$, where $\overline{\red}_{n+1}(\pset{P})$ is understood as $\pset{P}$. Then $G(\overline{\red}_{1}(\pset{P})) = G(\pset{P})$. 
\end{proposition}

\begin{proof}
Repeated use of Theorem~\ref{thm:reduction} implies
$$G(\pset{P}) = G(\red_n(\pset{P})) = G(\overline{\red}_{n-1}(\pset{P})) = \cdots = G(\overline{\red}_1(\pset{P})),$$
and the conclusion follows. 
\end{proof}


Proposition~\ref{prop:left} holds because after every reduction $G(\overline{\red}_i(\pset{P}))$ remains the same as the chordal graph $G(\pset{P})$, and thus the hypotheses of Theorem~\ref{thm:reduction} remain satisfied. If we weaken the condition $\supp(T_i) = \supp(\overline{\red}_{i+1}(\pset{P})^{(i)})$ in Proposition~\ref{prop:left} to $\supp(T_i) \subset \supp(\overline{\red}_{i+1}(\pset{P})^{(i)})$, then in general we will not have
$$G(\overline{\red}_1(\pset{P})) \subset \cdots \subset G(\overline{\red}_{n-1}(\pset{P})) \subset G(\red_n(\pset{P})) \subset G(\pset{P}),$$
as shown by the following example (though the last inclusion always holds because $G(\pset{P})$ is chordal). 

\begin{example}\label{ex:successiveInc}\rm
Let us continue with Example~\ref{ex:oneReduction} with $\pset{P}$ and $\pset{Q} = \red_5(\pset{P})$, where $G(\pset{Q}) \subsetneq G(\pset{P})$. Take 
\begin{equation*}
 T_4 = \prem(x_4^3+x_3, x_4^2+x_2) = -x_2x_4+x_3, \quad
\pset{R}_4 = \{\prem(x_4^2+x_2, -x_2x_4+x_3)\} = \{x_3^2-x_2^3\},
\end{equation*}
then 
$$\pset{Q}' := \overline{\red}_{4}(\pset{P}) = \{x_2+x_1, x_3+x_1, x_3^2-x_2^3, x_3, -x_2x_4+x_3, x_5+x_2\}.$$
The associated graph $G(\pset{Q}')$ is shown below. Note that $G(\pset{Q}') \not \subset G(\pset{Q})$ but $G(\pset{Q}') \subset G(\pset{P})$. 
    \begin{figure}[ht]
      \centering
\includegraphics[width=3cm,keepaspectratio]{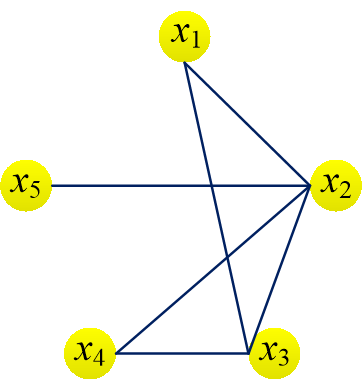}
      \caption{The associated graph $G(\pset{Q}')$ in Example~\ref{ex:successiveInc}}
      \label{fig:needCompletion}
    \end{figure}
\end{example}

Despite of this example where successive inclusions of the associated graphs in the reduction chain does not hold, it can be proved that for each $i=n, \ldots, 1$, $G(\overline{\red}_i(\pset{P}))$ is a subgraph of the original chordal graph $G(\pset{P})$.

\begin{lemma}\label{prop:subgraph}
  Let $\pset{P} \subset \kx$ be a chordal polynomial set with $x_1 < \cdots < x_n$ as one perfect elimination ordering of $G(\pset{P})$ and $\overline{\red}_i(\pset{P})$ be as defined in \eqref{eq:redBar} for $i=n, \ldots, 1$. Then for each $i=n, \ldots, 1$ and any two variables $x_p$ and $x_q$, if there exists an integer $k$ such that $x_p, x_q \in \supp(\overline{\red}_i(\pset{P})^{(k)})$, then $(x_p, x_q) \in G(\pset{P})$. 
\end{lemma}

\begin{proof}
  We induce on the integer $i$. In the case $i=n$, from the proof of Theorem~\ref{thm:reduction} one can easily find that the conclusion holds . Now suppose that the conclusion holds for $i=j~(\leq n)$, and next we prove that it also holds for $i=j-1$, namely for any $x_p$ and $x_q$, if there exists $k~(\geq p, q)$ such that $x_p, x_q \in \supp(\overline{\red}_{j-1}(\pset{P})^{(k)})$, then $(x_p, x_q) \in G(\pset{P})$. 

  Since $\overline{\red}_{j-1}(\pset{P}) = \red_{j-1}(\overline{\red}_j(\pset{P}))$, by \eqref{eq:reduction} we consider the following three cases of $k$.

  (1) If $k > j-1$, then $\overline{\red}_{j-1}(\pset{P})^{(k)} = \overline{\red}_j(\pset{P})^{(k)}$, and thus $x_p, x_q \in \supp(\overline{\red}_j(\pset{P})^{(k)})$. By the inductive assumption we have $(x_p, x_q) \in G(\pset{P})$.
  
  (2) If $k = j-1$, then $x_p, x_q \in \supp(T_{j-1}) \subset \supp(\overline{\red}_j(\pset{P})^{(j-1)})$, and thus by the inductive assumption we have $(x_p, x_q) \in G(\pset{P})$.

  (3) If $k < j-1$, then there exists a polynomial set $\tilde{\pset{R}}$ such that $\overline{\red}_{j-1}(\pset{P})^{(k)} = \overline{\red}_j(\pset{P})^{(k)} \cup \tilde{\pset{R}}^{(k)}$ and $\supp(\tilde{\pset{R}}) \subset \supp(\overline{\red}_j(\pset{P})^{(j-1)}) \setminus \{x_{j-1}\}$. 

(3.1) If $\tilde{\pset{R}} = \emptyset$, then $x_p, x_q \in \supp(\overline{\red}_j(\pset{P})^{(k)})$, and by the inductive assumption we know that $(x_p, x_q) \in G(\pset{P})$. 

(3.2) If $\tilde{\pset{R}} \neq \emptyset$, then $x_k \in \supp(\tilde{\pset{R}}^{(k)}) \subset \supp(\overline{\red}_j(\pset{P})^{(j-1)})$. Next we consider the following three cases. (3.2.1) $x_p, x_q \in \supp(\overline{\red}_j(\pset{P})^{(k)})$: with the same argument as in (a) we know that $(x_p, x_q) \in G(\pset{P})$. (3.2.2) $x_p, x_q \in \supp(\tilde{\pset{R}}^{(k)}) \subset \supp(\overline{\red}_j(\pset{P})^{(j-1)})$: by the induction assumption we know that $(x_p, x_q) \in G(\pset{P})$. (3.2.3) $x_p \in \supp(\overline{\red}_j(\pset{P})^{(k)})$ and $x_q \in \supp(\tilde{\pset{R}}^{(k)})$: Since $x_p, x_k \in \supp(\overline{\red}_j(\pset{P})^{(k)})$, by the induction assumption we have $(x_p, x_k) \in G(\pset{P})$; since $x_q, x_k \in \supp(\overline{\red}_j(\pset{P})^{(j-1)})$, by the induction assumption we have $(x_q, x_k) \in G(\pset{P})$. Then by the chordality of $\pset{P}$, $(x_p, x_k) \in G(\pset{P})$ and $(x_q, x_k) \in G(\pset{P})$ imply that $(x_p, x_q) \in G(\pset{P})$.

This ends the proof. 
\end{proof}

\begin{theorem}\label{thm:subgraph}
  Let $\pset{P} \subset \kx$ be a chordal polynomial set with $x_1 < \cdots < x_n$ as one perfect elimination ordering of $G(\pset{P})$ and $\overline{\red}_i(\pset{P})$ be as defined in \eqref{eq:redBar} for $i=n, \ldots, 1$. Then for each $i=n, \ldots, 1$, $G(\overline{\red}_i(\pset{P})) \subset G(\pset{P})$.
\end{theorem}

\begin{proof}
  By the construction of $\overline{\red}_i(\pset{P})$, we know that all the vertices of $G(\overline{\red}_i(\pset{P}))$ are also vertices of $G(\pset{P})$. For each edge $(x_p, x_q) \in G(\overline{\red}_i(\pset{P}))$, there exists an integer $k~(p, q \leq k \leq n)$ and a polynomial $P$ such that $x_p, x_q \in \supp(P)$ and $P \in \overline{\red}_i(\pset{P})^{(k)}$. Then by Lemma~\ref{prop:subgraph}, we know that $(x_p, x_q) \in G(\pset{P})$, and thus $G(\overline{\red}_i(\pset{P})) \subset G(\pset{P})$. 
\end{proof}

\begin{corollary}\label{cor:subgraph2}
  Let $\pset{P} \subset \kx$ be a chordal polynomial set with $x_1 < \cdots < x_n$ as one perfect elimination ordering of $G(\pset{P})$ and $\overline{\red}_i(\pset{P})$ be as defined in \eqref{eq:redBar} for $i=n, \ldots, 1$. If $\pset{T} := \overline{\red}_1(\pset{P})$ does not contain any nonzero constant, then $\pset{T}\setminus \{0\}$ forms a triangular set such that $G(\pset{T}) \subset G(\pset{P})$. 
\end{corollary}

Corollary~\ref{cor:subgraph2} tells us that under the conditions that the input polynomial set is chordal and the variable ordering is one perfect elimination ordering, the associated graph of one specific triangular set computed in any algorithm for triangular decomposition in top-down style with reduction satisfying the conditions \eqref{eq:mapping} and \eqref{eq:reduction} is a subgraph of the associated graph of the input polynomial set. In fact, this triangular set is usually the ``main branch'' in the triangular decomposition in the sense that other branches are obtained by adding additional constrains in the splitting in the process of triangular decomposition. 

Note that in the case when the input polynomial set $\pset{P}$ is not chordal, a process of chordal completion can be carried out on $G(\pset{P})$ to generate a chordal graph (in the worst case this chordal completion results in a complete graph which is trivially chordal). After this chordal completion the conditions of Corollary~\ref{cor:subgraph2} will be satisfied.

The chordality of any triangular set other than the specific one above in a triangular decomposition computed by an algorithm in top-down style is dependent on the splitting strategy in the algorithm. In the following sections, we study several specific algorithms for triangular decomposition in top-down style and prove that the associated graphs of all the polynomial sets in the decomposition process of these algorithms are subgraphs of the associated graph of a chordal input polynomial set.

\section{Wang's method for triangular decomposition in top-down style}
\label{sec:wang}

A simply-structured algorithm was proposed by Wang for triangular decomposition in top-down style in 1993 \cite{w93e}, which is referred to as Wang's method in the literature (see. e.g., \cite{a99ta}). Next the chordality of polynomial sets in the decomposition process of Wang's method is studied. 

\subsection{Wang's method revisited}
\label{sec:reformulation}

For the self-containness of this paper, Wang's method for triangular decomposition is outlined in Algorithm~\ref{alg:wang} below. In this algorithm and those to follow, the data structure $(\pset{P}, \pset{Q}, k)$ is used to represent two polynomial sets $\pset{P}$ and $\pset{Q}$ such that $\#\pset{P}^{(i)} =0$ or $1$ for $i=k+1, \ldots, n$. For a set $\Phi$ consisting of tuples in the form $(\pset{P}, \pset{Q}, i)~(i=0, \ldots, n)$, denote $\Phi^{(k)} := \{(\pset{P}, \pset{Q}, i)\in \Phi|\, i = k\}$. The subroutine $\pop(\pset{S})$ returns an element from a set $\pset{S}$ and then removes it from $\pset{S}$. 

\begin{algorithm}[!ht]
\caption{Wang's method for triangular decomposition $\Psi\!:=\!\algwang(\pset{F})$}
\label{alg:wang}

\small

\KwIn{$\pset{F}$, a polynomial set in $\kx$}

\KwOut{$\Psi$, a set of finitely many triangular systems which form a triangular decomposition of $\pset{F}$}

\BlankLine
$\Phi := \{(\pset{F}, \emptyset, n)\}$\;
$\Psi := \emptyset$\;

\For{$k = n, \ldots, 1$}
{  \While{$\Phi^{(k)} \neq \emptyset$}
   {
      $(\pset{P}, \pset{Q}, k) := \pop(\Phi^{(k)})$\;

      \eIf{$\#\pset{P}^{(k)} > 1$}
      {$T:=$ a polynomial in $\pset{P}^{(k)}$ with minimal degree in $x_k$\; \label{line:minimal}
        $\pset{P}' := \pset{P} \setminus \pset{P}^{(k)} \cup \{T\} \cup \{\prem(P, T): P \in \pset{P}^{(k)} \setminus \{T\}\}$\;
        $\Phi := \Phi \cup \{(\pset{P} \setminus \{T\} \cup \{\ini(T), \tail(T)\}, \pset{Q}, k)\} \cup  \{(\pset{P}', \pset{Q} \cup \{\ini(T)\}, k)\}$ \label{line:wang-end}\;
      }
      {  $\Phi := \Phi \cup \{(\pset{P}, \pset{Q}, k-1)\}$}
   }
}

\For{$(\pset{P}, \pset{Q}, 0) \in \Phi^{(0)}$}
{
   \If{$\pset{P}^{(0)} \setminus \{0\} = \emptyset$}
   {
      $\Psi := \Psi \cup \{(\pset{P} \setminus \{0\}, \pset{Q})\}$\;
   }
}

\Return $\Psi$\;
\end{algorithm} 


The decomposition process in Wang's method (Algorithm~\ref{alg:wang}) applied to $\pset{F}$ can be viewed as a binary tree with its root as $(\pset{F}, \emptyset, n)$. The nodes of this binary tree are all the tuples $(\pset{P}, \pset{Q}, k)$ picked from $\Phi$, and each node $(\pset{P}, \pset{Q}, k)$ has two child nodes $(\pset{P}', \pset{Q}', k)$ and $(\pset{P}'', \pset{Q}'', k)$, where 
\begin{equation}
  \label{eq:wang-split}
  \begin{split}
    \pset{P}' &\!:=\! \pset{P} \!\setminus\! \pset{P}^{(k)} \!\cup\! \{T\} \!\cup\! \{\prem(P, T): P \in \pset{P}^{(k)}\setminus \{T\}\}, ~~~ \pset{Q}' := \pset{Q} \cup \{\ini(T)\}, \\
    \pset{P}'' &\!:=\! \pset{P} \!\setminus\! \{T\} \!\cup\! \{\ini(T), \tail(T)\}, \qquad\quad\quad\qquad\qquad\quad~~    \pset{Q}'' := \pset{Q},
  \end{split}
\end{equation}
with $T$ as a polynomial in $\pset{P}^{(k)}$ with minimal degree in $x_k$. In fact, the left child node $(\pset{P}', \pset{Q}', k)$ corresponds to the case when $\ini(T)\neq 0$ and thus reduction of $\pset{P}^{(k)}$ are performed with respect to $T$; while the right child node $(\pset{P}'', \pset{Q}'', k)$ corresponds to the case $\ini(T) = 0$, where $T$ is replaced by $\ini(T)$ and $\tail(T)$. 

The binary decomposition tree for Wang's method and the splitting at one node are illustrated in Figures~\ref{fig:decTree}.

    \begin{figure}[ht]
      \centering
\includegraphics[width=15cm,keepaspectratio]{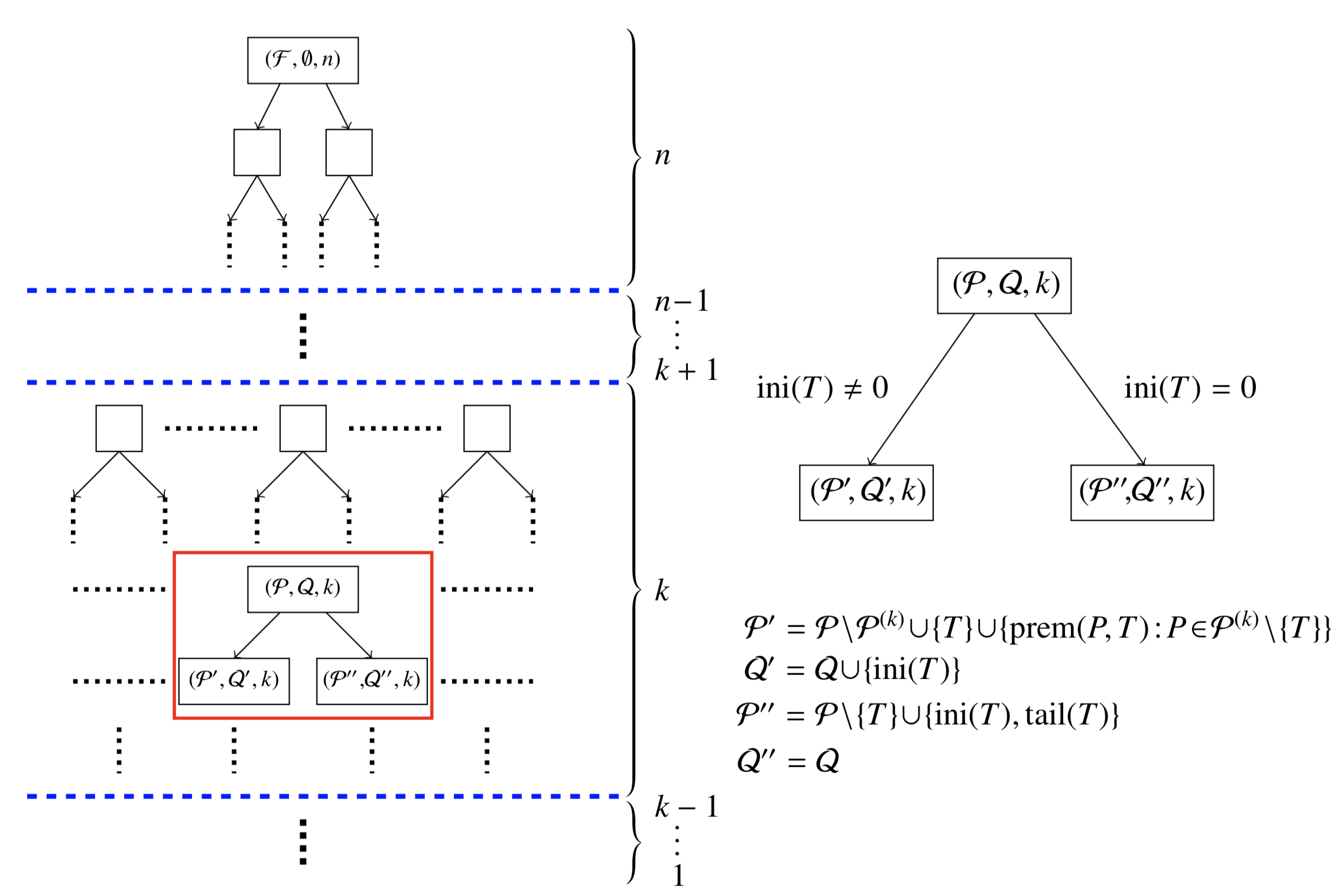}
\caption{Binary decomposition tree for Wang's method}
\label{fig:decTree}
\end{figure}

\subsection{Chordality of polynomial sets in Wang's method}
\label{sec:graphs}

With a chordal polynomial set as the input of Wang's method, the relationships between the associated graphs of the polynomial sets in the left nodes and that of the input polynomial set and between the associated graphs of the polynomial sets in the right child nodes and those in the parent nodes are clarified in the following propositions.

\begin{proposition}\label{prop:wang-left}
  Let $\pset{F} \subset \kx$ be a chordal polynomial set with $x_1 < \cdots < x_n$ as one perfect elimination ordering of $G(\pset{F})$, $(\pset{P}, \pset{Q}, k)$ be any node in the binary decomposition tree of $\algwang(\pset{F})$ such that $G(\pset{P}) \subset G(\pset{F})$, $T$ be a polynomial in $\pset{P}$ with minimal degree in $x_k$, and $\pset{P}'$ be as defined in \eqref{eq:wang-split}. Then $G(\pset{P}') \subset G(\pset{F})$.
\end{proposition}

\begin{proof}
  Clearly $\supp(\pset{P}') \subset \supp(\pset{P}) \subset \supp(\pset{F})$, and it suffices to prove that for any edge $(x_p, x_q) \in G(\pset{P}')$, we have $(x_p, x_q) \in G(\pset{F})$. 

Denote $\pset{R} := \{\prem(P, T):\, P\in \pset{P}^{(k)}\setminus \{T\}\}$. Then by definition $\supp(\pset{R}^{(i)}) \subset \supp(\pset{P}^{(k)})$ for any $i=1, \ldots, k$. Furthermore, the following relationships hold: $\pset{P}'^{(i)} = \pset{P}^{(i)} \cup \pset{R}^{(i)}$ for $i=1, \ldots, k-1$ and $\pset{P}'^{(k)} = \{T\} \cup \pset{R}^{(k)}$. For any edge $(x_p, x_q) \in G(\pset{P}')$, there exist an integer $j~(p, q \leq j \leq k)$ and a polynomial $P \in \pset{P}'^{(j)}$ such that $x_p, x_q\in \supp(P)$. 

In the case when $j=k$, we have $(x_p, x_q)\in G(\pset{P}'^{(k)})$. Then $x_p, x_q \in \supp(T) \cup \supp(\pset{R}^{(k)}) \subset \supp(\pset{P}^{(k)})$ and thus $(x_p, x_k)$, $(x_q, x_k) \in G(\pset{P}^{(k)}) \subset G(\pset{F})$. By the chordality of $\pset{F}$, we have $(x_p, x_q) \in G(\pset{F})$.

In the case when $j<k$, we have $(x_p, x_q) \in G(\pset{P}'^{(j)})$ with $\pset{P}'^{(j)} = \pset{P}^{(j)} \cup \pset{R}^{(j)}$. If $P \in \pset{P}^{(j)}$, then it is obvious that $(x_p, x_q) \in G(\pset{P}^{(j)}) \subset G(\pset{F})$; otherwise if $P\in \pset{R}^{(j)}$, then $x_p, x_q \in \supp(R^{(j)}) \subset \supp(\pset{P}^{(k)})$ and thus $(x_p, x_k), (x_q, x_k) \in G(\pset{P}) \subset G(\pset{F})$. Then the chordality of $\pset{F}$ implies $(x_p, x_q) \in G(\pset{F})$.
\end{proof}

\begin{proposition}\label{prop:wang-right}
Let $(\pset{P}, \pset{Q}, k)$ be any node in the binary decomposition tree of $\algwang(\pset{F})$, $T$ be a polynomial in $\pset{P}^{(k)}$ with minimal degree in $x_k$, and $\pset{P}''$ be defined as in \eqref{eq:wang-split}. Then $G(\pset{P}'') \subset G(P)$. In particular, if $\supp(\tail(T)) = \supp(T)$, then $G(\pset{P}'') = G(\pset{P})$. 
\end{proposition}

\begin{proof}
  Since $\pset{P}''$ is constructed by replacing $T$ in $\pset{P}$ with $\ini(T)$ and $\tail(T)$, we only need to study the differences between $G(\pset{P})$ and $G(\pset{P}'')$ caused by this replacement. First, by $\supp(\ini(T)) \cup \supp(\tail(T)) \subset \supp(T)$ we have $\supp(\pset{P}'') \subset \supp(\pset{P})$. Second, for any edge $(x_p, x_q)$ in $G(\ini(T))$ or in $G(\tail(T))$, we know that $(x_p, x_q) \in G(T)$, which means that all the edges of $G(\pset{P}'')$ are also edges of $G(\pset{P})$. Therefore, $G(\pset{P}'') \subset G(\pset{P})$.

In particular, if $\supp(\tail(T)) = \supp(T)$, then $\supp(\ini(T)) \cup \supp(\tail(T)) = \supp(T)$ and any edge $(x_p, x_q) \in \supp(T)$ is also contained in $G(\tail(T))$, and thus $G(\pset{P}'') = G(\pset{P})$. 
\end{proof}

\begin{example}\label{ex:wang-right}\rm
  Let 
  \begin{equation*}
    \begin{split}
\pset{P}_1 &= [x_1+x_2, x_1+x_3, x_2+x_3, x_4^3+x_1, x_3x_4^2+x_3+x_4], \\
\pset{P}_2 &= [x_1+x_2, x_1+x_3, x_2+x_3, x_4^3+x_1, x_3x_4^2+x_4].
    \end{split}
  \end{equation*}
Then $G(\pset{P}_1) = G(\pset{P}_2)$ is shown in Figure~\ref{fig:tail} below (left). Let $\pset{P}_1''$ and $\pset{P}_2''$ be constructed from $\pset{P}_1$ and $\pset{P}_2$ with respect to $x_4$ respectively. Then $x_3x_4^2+x_3+x_4$ and $x_3x_4^2+x_4$ are chosen as $T$ respectively and
  \begin{equation*}
\pset{P}_1'' = [x_1+x_2, x_1+x_3, x_2+x_3, x_3, x_4^3+x_1, x_3+x_4], 
\pset{P}_2'' = [x_1+x_2, x_1+x_3, x_2+x_3, x_3, x_4^3+x_1, x_4].
  \end{equation*}
One may check that $G(\pset{P}_1'') = G(\pset{P}_1)$ while $G(\pset{P}_2'') \neq G(\pset{P}_2)$, with $G(\pset{P}_2'')$ shown in Figure~\ref{fig:tail} below (right).

    \begin{figure}[ht]
      \centering
\includegraphics[width=3cm,keepaspectratio]{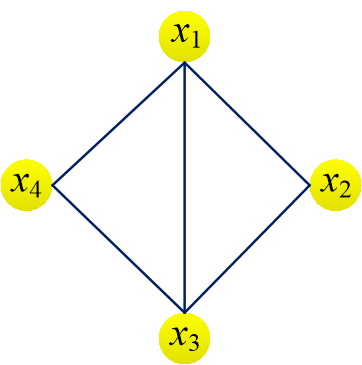}~\qquad \qquad
\includegraphics[width=3cm,keepaspectratio]{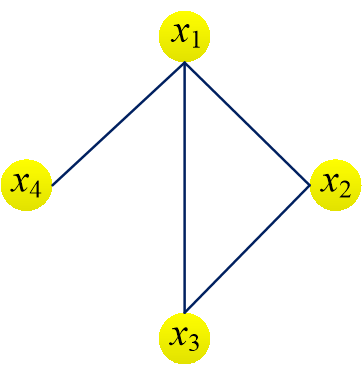}      
      \caption{The associated graphs $G(\pset{P}_1) = G(\pset{P}_2) = G(\pset{P}_1'')$ (left) and $G(\pset{P}_2'')$ (right) in Example~\ref{ex:wang-right}}
      \label{fig:tail}
\end{figure}
\end{example}

Next we prove that with a chordal input polynomial set, the polynomial sets in all the nodes of the decomposition tree of Wang's method, and thus all the computed triangular sets, have associated graphs which are subgraphs of that of the input polynomial set. 

  \begin{theorem}\label{thm:wang}
Let $\pset{F} \subset \kx$ be a chordal polynomial set with $x_1 < \cdots < x_n$ as one perfect elimination ordering of $G(\pset{F})$. Then for any node $(\pset{P}, \pset{Q}, k)$ in the binary decomposition tree of $\algwang(\pset{F})$, we have $G(\pset{P}) \subset G(\pset{F})$ and $G(\pset{Q}) \subset G(\pset{F})$.
  \end{theorem}

  \begin{proof}
    We induce on the depth $d$ of $(\pset{P}, \pset{Q}, k)$ in the binary decomposition tree. When $d=0$, the conclusion naturally holds. Now assume that for any node $(\tilde{\pset{P}}, \tilde{\pset{Q}}, k)$ of depth $d$ in the decomposition tree, we have $G(\tilde{\pset{P}}) \subset G(\pset{F})$ and $G(\tilde{\pset{Q}}) \subset G(\pset{F})$. Let $(\pset{P}, \pset{Q}, k)$ be of depth $d+1$ and $(\tilde{\pset{P}}, \tilde{\pset{Q}}, k)$ be its parent node of depth $d$ in the decomposition tree. If $(\pset{P}, \pset{Q}, k)$ is a left child node, then $G(\pset{P}) \subset G(\pset{F})$ by Proposition~\ref{prop:wang-left} and $G(\pset{Q}) \subset G(\pset{F})$, for $\pset{Q} = \tilde{\pset{Q}} \cup \{\ini(T)\}$ with $G(\tilde{Q}) \subset G(\pset{F})$ and some polynomial $T \in \tilde{\pset{P}}$, where $G(\tilde{\pset{P}}) \subset G(\pset{F})$; if it is a right node, $G(\pset{P}) \subset G(\tilde{\pset{P}}) \subset G(\pset{F})$ by Proposition~\ref{prop:wang-right} and $G(\pset{Q}) = G(\tilde{\pset{Q}}) \subset G(\pset{F})$. This ends the inductive proof.

  \end{proof}

  \begin{corollary}\label{cor:wang-ts}
    Let $\pset{F} \subset \kx$ be a chordal polynomial set with $x_1 < \cdots < x_n$ as one perfect elimination ordering of $G(\pset{F})$ and $\pset{T}_1, \ldots, \pset{T}_r$ be the triangular sets computed by $\algwang(\pset{F})$. Then $G(\pset{T}_i) \subset G(\pset{F})$ for $i=1, \ldots, r$. 
  \end{corollary}

  \begin{proof}
    Straightforward from Theorem~\ref{thm:wang} with the fact that each triangular set $\pset{T}_i$ is from some node in the decomposition tree. 
  \end{proof}

\subsection{An illustrative example}
\label{sec:wangExample}

Here we illustrate the changes of the associated graphs of polynomial sets computed in the triangular decomposition via Wang's method applied to
\begin{equation}
  \label{eq:illus}
\pset{F} = \{x_2+x_1+2, (x_2+2)x_3+x_1, (x_3+x_2)x_4 + x_3 -1, x_4+x_2\} \subset \qnum[x_1, x_2, x_3, x_4]
\end{equation}
for the variable ordering $x_1 < x_2 < x_3 < x_4$. The associated graph $G(\pset{F})$ is shown in Figure~\ref{fig:exa1}, and one can check that $G(\pset{F})$ is chordal with $x_1 < x_2 < x_3 < x_4$ as one perfect elimination ordering. 

\begin{figure}[ht]
      \centering
\includegraphics[width=3cm,keepaspectratio]{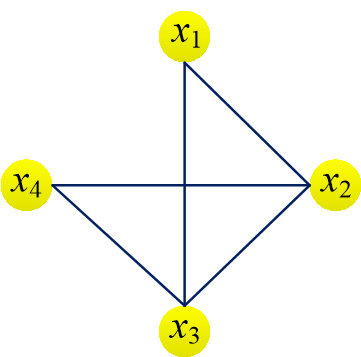}    
      \caption{The associated graph $G(\pset{F})$ with $\pset{F}$ from \eqref{eq:illus}}
      \label{fig:exa1}
    \end{figure}

First $T = (x_3+x_2)x_4+x_3-1$ is chosen as the polynomial in $\pset{F}^{(4)}$ with minimal degree in $x_4$, then the right node $(\pset{F}', \emptyset, 4)$ for the case $\ini(T) = x_3+x_2 = 0$ is added to $\Phi$ for further computation, where $\pset{F}' = \{x_2+x_1+2, (x_2+2)x_3+x_1, (x_3+x_2), x_3-1, x_4+x_2\}$. The psuedo division of $x_4+x_2 \in \pset{F}^{(4)}$ with respect to $T$ results in 
$$\pset{P} = \{x_2+x_1+2, (x_2+2)x_3+x_1, (x_2-1)x_3 + x_2^2 +1, (x_3+x_2)x_4 + x_3 -1\},$$
and thus the left child node is $(\pset{P}, \{x_3+x_2\}, 3)$.

Next $T' = (x_2+2)x_3+x_1$ is chosen as the polynomial in $\pset{P}^{(3)}$ with minimal degree in $x_3$, then the right node $(\pset{F}'', \{x_3+x_2\}, 3)$ is added to $\Phi$, where $\pset{F}'' = \{x_1, x_2+x_1+2, x_2+2, (x_2-1)x_3 + x_2^2 +1, (x_3+x_2)x_4 + x_3 -1\}$, and the pseudo division of $(x_2-1)x_3 + x_2^2 +1 \in \pset{P}^{(3)}$ with respect to $T'$ results in 
\begin{equation*}
\pset{P}' = \{x_2+x_1+2, x_2^3 + 2x_2^2 - (x_1-1)x_2 + x_1+2, (x_2+2)x_3+x_1, (x_3+x_2)x_4 + x_3 -1\}      
\end{equation*}
and thut the left node is $(\pset{P}', \{x_3+x_2, x_2+2\}, 2)$.

At this step $T'' = x_2+x_1+2$ is chosen as the polynomial in $\pset{P}'^{(2)}$ with minimal degree in $x_2$, then with $\ini(T'') = 1$ the right node implies conflicts and nothing is added to $\Phi$, and the pseudo-division of $x_2^3 + 2x_2^2 - (x_1-1)x_2 + x_1+2 \in \pset{P}'^{(2)}$ with respect to $T''$ results in the first triangular set 
\begin{equation}
  \label{eq:illus-T1}
\pset{T}_1 \!=\! [-x_1^3 + x_1^2 + 14x_1 + 16, x_2+x_1+2, (x_2+2)x_3+x_1, (x_3+x_2)x_4 + x_3 -1].
\end{equation}
With similar treatments on  $(\pset{F}', \emptyset, 4)$ and $(\pset{F}'', \{x_3+x_2\}, 3)$ in $\Phi$, the other two triangular sets 
\begin{equation}
  \label{eq:illus-T23}
    \pset{T}_2 = [x_1+1, x_2+1, x_3-1, x_4+x_2], ~~
    \pset{T}_3 = [x_1, x_2+2, (x_2-1)x_3+x_2^2+1, (x_3+x_2)x_4+x_3-1]
\end{equation}
are computed. 

The associated graphs of all these three computed triangular sets are shown in Figure~\ref{fig:exa2}. One can find that the associated graphs $G(\pset{F})$ and $G(\pset{T}_1)$ are the same, while $G(\pset{T}_2)$ and $G(\pset{T}_3)$ are strict subgraphs of $G(\pset{F})$.

\begin{figure}[ht]
      \centering
\includegraphics[width=3cm,keepaspectratio]{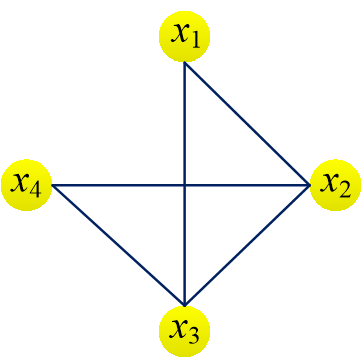}~\qquad\quad
\includegraphics[width=3cm,keepaspectratio]{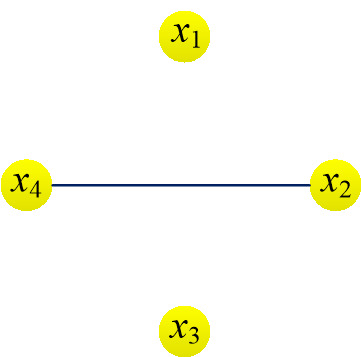}~\qquad\quad
\includegraphics[width=3cm,keepaspectratio]{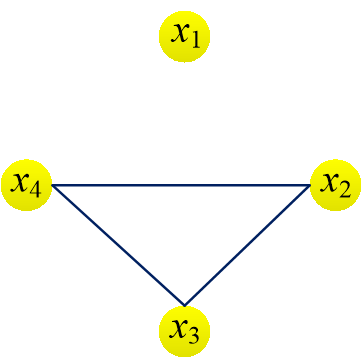}
      \caption{The associated graphs $G(\pset{T}_1)$, $G(\pset{T}_2)$, and $G(\pset{T}_3)$ with $\pset{T}_1$ from \eqref{eq:illus-T1} and $\pset{T}_2, \pset{T}_3$ from \eqref{eq:illus-T23}}
      \label{fig:exa2}
    \end{figure}

\section{Subresultant-based algorithm for triangular decomposition in top-down style}
\label{sec:subres}

\subsection{Subresultant-based algorithm for triangular decomposition revisited}
\label{sec:StateSubres}

First we reformulate and modify the algorithm in top-down style for triangular decomposition based on computation of subresultant regular subchains due to Wang as Algorithm~\ref{alg:srs}. One is referred to \cite[Section~2.4]{W2001E} for the details of this algorithm and Algorithm~TriSerS therein for the original description. In Algorithm~\ref{alg:srs} below, the subroutine $\srs(T_1, T_2)$ returns the subresultant regular subchain of $T_1$ and $T_2$ with respect to $\lv(T_2)$ when $\deg(T_1, \lv(T_2)) \geq \ldeg(T_2)$.


\begin{algorithm}[!ht]
  \caption{Subresultant-based algorithm for triangular decomposition $~~~~~~~~~~~~~~~~~~~~~~~~\Psi:=\algSub(\pset{F})$}
\label{alg:srs}

\small

\KwIn{$\pset{F}$, a polynomial set in $\kx$}

\KwOut{$\Psi$, a set of finitely many triangular systems which form a triangular decomposition of $\pset{F}$}

\BlankLine
$\Phi := \{(\pset{F}, \emptyset, n)\}$\;
$\Psi := \emptyset$\;

\For{$k = n, \ldots, 1$}
{  \While{$\Phi^{(k)} \neq \emptyset$}
   {
      $(\pset{P}, \pset{Q}, k) := \pop(\Phi^{(k)})$\;

      \eIf{$\#\pset{P}^{(k)} > 1$}
      {
         $T_2:=$ a polynomial in $\pset{P}^{(k)}$ with minimal degree in $x_k$ \label{line:srs-minimal}\; 
         $\Phi := \Phi \cup \{(\pset{P} \setminus \{T_2\} \cup \{\ini(T_2), \tail(T_2)\}, \pset{Q}, k)\}$ \label{line:Pprime}\;
         $T_1 := \pop(\pset{P}^{(k)} \setminus \{T_2\})$ \label{line:srs-another}\;
         $(H_2, \ldots, H_r) := \srs(T_1, T_2)$\;
         \For{$i=2, \ldots, r-1$}
         {
            $\Phi := \Phi \cup \{(\pset{P}\setminus \{T_1, T_2\} \cup \{H_i, \ini(H_{i+1}), \ldots, \ini(H_{r})\}, \pset{Q}\cup \{\ini(T_2)\}\cup \{\ini(H_i)\}, k)\}$\label{line:sub-split}\;
         }
         $\Phi := \Phi \cup \{(\pset{P} \setminus \{T_1, T_2\} \cup \{H_r\}, \pset{Q}\cup \{\ini(T_2)\}\cup\{\ini(H_r)\}, k) \}$ \label{line:end}\; 
       }
       {  $\Phi := \Phi \cup \{(\pset{P}, \pset{Q}, k-1)\}$ \label{line:replace}}
   }
}

\For{$(\pset{P}, \pset{Q}, 0) \in \Phi^{(0)}$}
{
   \If{$\pset{P}^{(0)} \setminus \{0\} = \emptyset$}
   {
      $\Psi := \Psi \cup \{(\pset{P} \setminus \{0\}, \pset{Q})\}$\;
   }
}

\Return $\Psi$\;
\end{algorithm}

We make one modification in Algorithm~\ref{alg:srs} against the original algorithm: in line~\ref{line:end} the second polynomial set in the node to be added to $\Phi$ is $\pset{Q}_r := \pset{Q}\cup \{\ini(T_2)\}\cup\{\ini(H_r)\}$, but in the original algorithm it is $\{\prem(Q, T_2): Q \in \pset{Q}_r\}$. This modification is made because pseudo division here is indeed against the essence of triangular decomposition in top-down style: $Q \in \pset{Q}$ may involve a variable $x_l$ such that $x_l > x_k$, and thus in this case replacing $Q$ with $\prem(Q, T_2)$ changes polynomials in $\pset{Q}^{(l)}$ with $l > k$, which conflicts with condition~(b) for an algorithm for triangular decomposition to be in top-down style. In fact, this operation may indeed result in a polynomial set whose associated graph is not a subgraph of the input chordal graph. In the original algorithm, the computation $\prem(Q, T_2)$ for $Q \in \pset{Q}_r$ is for ensuring the final triangular sets to be \emph{fine} (see \cite[Page~23]{W2001E} for the definition), and removing this in Algorithm~\ref{alg:srs} does not affact the correctness of the algorithm. 

Comparing the descriptions of Algorithms~\ref{alg:srs} and \ref{alg:wang}, one can find that these two algorithms are structurally similar and that Algorithm~\ref{alg:srs} only replace lines \ref{line:minimal}--\ref{line:wang-end} in Algorithm~\ref{alg:wang} by the new lines \ref{line:srs-minimal}--\ref{line:end}. Furthermore, the decomposition tree of Algorithm~\ref{alg:srs} can be constructed in the following way. This decomposition tree is rooted at $(\pset{F}, \emptyset, n)$ where $\pset{F}$ is the input polynomial set, and any node $(\pset{P}, \pset{Q}, k)$ in this decomposition tree has the following child nodes: $(\pset{P}', \pset{Q}', k)$ corresponding to line~\ref{line:Pprime}, $(\pset{P}_i, \pset{Q}_i, k)$ for $i=2, \ldots, r-1$ corresponding to line~\ref{line:sub-split}, and $(\pset{P}_r, \pset{Q}_r, k)$ corresponding to line~\ref{line:end}, where $\pset{P}', \pset{Q}', \pset{P}_i$ and $\pset{Q}_i$ are defined as:
\begin{equation}
  \label{eq:srs-def}
  \begin{split}
    \pset{P}' &:= \pset{P} \setminus \{T_2\} \cup \{\ini(T_2), \tail(T_2)\},\\
    \pset{Q}' &:= \pset{Q}, \\
    \pset{P}_{i} &: = \left\{
      \begin{tabular}[l]{ll}
        $\pset{P}\setminus \{T_1, T_2\} \cup \{H_i, \ini(H_{i+1}), \ldots, \ini(H_{r})\}$, & $i=2, \ldots, r-1$, \\
        $\pset{P} \setminus \{T_1, T_2\} \cup \{H_r\}$, & $i=r$,
      \end{tabular}
    \right.\\
    \pset{Q}_i  &: = \pset{Q}\cup \{\ini(T_2)\}\cup \{\ini(H_i)\}, \qquad\qquad\qquad\qquad\quad~~\! i=2, \ldots, r.
        \end{split}
\end{equation}
Since the number of polynomials in $\{H_2, \ldots, H_r\}$ is dynamic, we will have a dynamic multi-branch decomposition tree for $\algSub(\pset{F})$ as illustrated in Figure~\ref{fig:decTreeSRS}. 

    \begin{figure}[ht]
      \centering
\includegraphics[width=10cm,keepaspectratio]{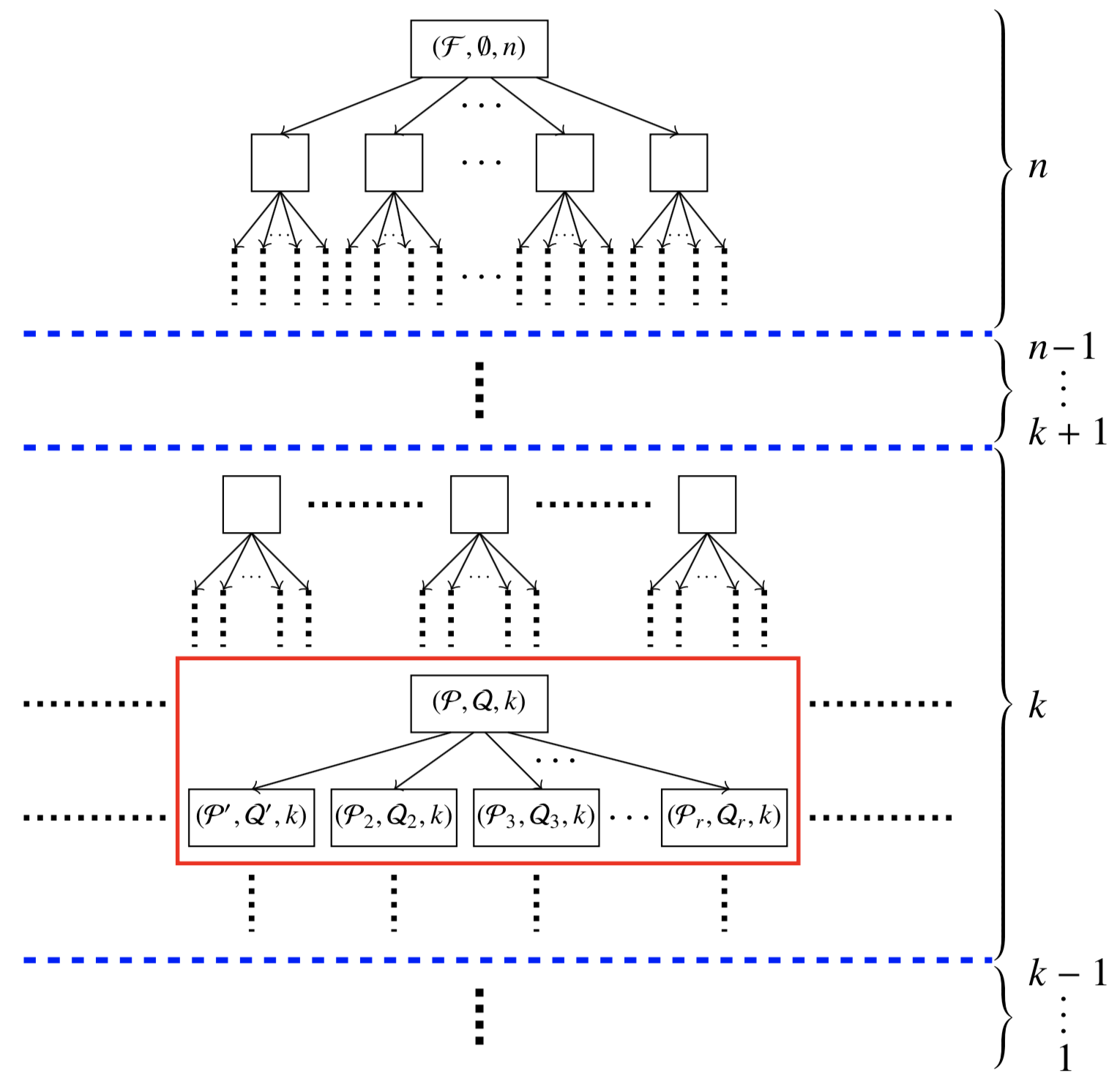}
\caption{Dynamic multi-branch decomposition tree of $\algSub(\pset{F})$}
\label{fig:decTreeSRS}
\end{figure}

\subsection{Chordality of polynomial sets in subresultant-based algorithm for triangular decomposition}
\label{sec:chordalSubres}


Next we show that chordality of polynomial sets is ``preserved'' with computation of the subresultant regular subchain, and thus the associated graphs of the two polynomial sets in each node of the decomposition tree of $\algSub(\pset{F})$ are subgraphs of the input chordal graph, which is the same conclusion as Theorem~\ref{thm:wang}.

\begin{proposition}
  \label{prop:td-srs}
  Let $\pset{F} \subset \kx$ be a chordal polynomial set with $x_1 < \cdots < x_n$ as one perfect elimination ordering of $G(\pset{F})$, and $(\pset{P}, \pset{Q}, k)$ be an arbitrary node in the decomposition tree of $\algSub(\pset{F})$ such that $G(\pset{P}) \subset G(\pset{F})$ and $\#\pset{P}^{(k)} > 1$. Let $T_2$ and $T_1$ be chosen from $\pset{P}$ as in lines~\ref{line:srs-minimal} and \ref{line:srs-another} in $\algSub(\pset{F})$, $H_2, \ldots, H_r$ be the subresultant regular subchain of $T_1$ and $T_2$ with respect to $x_k$, and $\pset{P}_i$ be as defined in \eqref{eq:srs-def} for $i=2, \ldots, r$. Then we have $G(\pset{P}_{i}) \subset G(\pset{F})$ for $i=2, \ldots, r$.
\end{proposition}

\begin{proof}
Since $H_2, \ldots, H_r$ form a subresultant regular subchain of $T_1$ and $T_2$ with $T_1, T_2 \in \pset{P}$, we know that $\supp(\ini(H_i)) \subset \supp(H_i) \subset \supp(T_1) \cup \supp(T_2) \subset \supp(\pset{P}) \subset \supp(\pset{F})$ for each $i=2, \ldots, r$ and thus $\supp(\pset{P}_{i}) \subset \supp(\pset{F})$. Next we show that for any edge $(x_p, x_q) \in G(\pset{P}_{i})$, we have $(x_p, x_q) \in G(\pset{F})$.

For each $i=2, \ldots, r$, if there exists a polynomial $T \in \pset{P} \setminus \{T_1, T_2\}$ such that $x_p, x_q \in \supp(T)$, then $(x_p, x_q) \in G(\pset{P}) \subset G(\pset{F})$ by the assumption. Else there exists a polynomial $T \in \{H_i, \ini(H_{i+1}), \ldots, \ini(H_{r})\}$ in the case of $2 \leq i <r$ or a polynomial $T = H_r$ in the case of $i=r$ such that $x_p, x_q \in \supp(T)$, and thus $x_p, x_q \in \supp(T_1) \cup \supp(T_2)$.

If $x_p, x_q \in \supp(T_1)$ or $x_p, x_q \in \supp(T_2)$, then clearly $(x_p, x_q) \in G(\pset{P}) \subset G(\pset{F})$. Else, without loss of generality, we assume that $x_p \in \supp(T_1)$ and $x_q \in \supp(T_2)$. Noting that $T_1, T_2 \in \pset{P}^{(k)}$, we know $x_k \in \supp(T_1)$ and $\supp(T_2)$. Then $(x_p, x_k), (x_q, x_k) \in G(\pset{P}) \subset G(\pset{F})$ and $x_p, x_q \leq x_k$. The chordality of $G(\pset{F})$ implies the inclusion $(x_p, x_q) \subset G(\pset{F})$. 
\end{proof}

\begin{theorem}
  \label{thm:chordalSRS}
Let $\pset{F} \subset \kx$ be a chordal polynomial set with $x_1 < \cdots < x_n$ as one perfect elimination ordering of $G(\pset{F})$. Then for any node $(\pset{P}, \pset{Q}, k)$ in the decomposition tree of $\algSub(\pset{F})$, we have $G(\pset{P}) \subset G(\pset{F})$ and $G(\pset{Q}) \subset G(\pset{F})$. 

\end{theorem}

\begin{proof}
With Propositions~\ref{prop:wang-right} and \ref{prop:td-srs} the proof of this theorem is almost the same as that of Theorem~\ref{thm:wang}. 
\end{proof}

  \begin{corollary}\label{cor:srs-ts}
    Let $\pset{F} \subset \kx$ be a chordal polynomial set with $x_1 < \cdots < x_n$ as one perfect elimination ordering of $G(\pset{F})$ and $\pset{T}_1, \ldots, \pset{T}_r$ be the triangular sets computed by $\algSub(\pset{F})$. Then $G(\pset{T}_i) \subset G(\pset{F})$ for $i=1, \ldots, r$. 
  \end{corollary}

  \begin{proof}
    Straightforward from Theorem~\ref{thm:chordalSRS}. 
  \end{proof}

\section{Subresultant-based algorithm for regular decomposition in top-down style}
\label{sec:regular}

In this section we study the change of associated graphs of polynomial sets in an algorithm for regular decomposition in top-down style using subresultant regular subchains. Due to the strict constraints on regular systems in regular decomposition, in this algorithm there exist interactions between polynomials from the two different sets representing equations and inequations. This interaction is mainly in the form of computing a subresultant regular subchain of a polynomial from the equation set and another polynomial from the inequation set. 

\subsection{Subresultant-based algorithm for regular decomposition revisited}
\label{sec:StateReg}

Similarly as in the previous sections, we first reformulate an algorithm in top-down style for regular decomposition due to Wang as Algorithm~\ref{alg:reg}, and interested readers are referred to \cite{w00c,W2001E} for the details. In this algorithm, the subroutine $\max(P, Q)$ (respectively $\min(P, Q)$) for two polynomials $P$ and $Q$ such that $\lv(P) = \lv(Q) = x_k$ returns the polynomial in $\{P, Q\}$ whose degree in $x_k$ is greater than or equal to (respectively is lower than) that of the other. 

\begin{algorithm}[!ht]
  \caption{Subresultant-based algorithm for regular decomposition    $~~~~~~~~~~~~~~~~~~~~~~~~~\Psi:=\algReg(\pset{F})$}
\label{alg:reg}

\small

\KwIn{$\pset{F}$, a polynomial set in $\kx$}

\KwOut{$\Psi$, a set of finitely many regular systems which form a regular decomposition of $\pset{F}$}

\BlankLine

[Replace line~\ref{line:replace} of Algorithm~\ref{alg:srs} by the following lines] 

         \If{$\pset{Q}^{(k)} = \emptyset$ \label{line:replace-start}}
         {
           $\Phi := \Phi \cup \{(\pset{P}, \pset{Q}, k-1)\}$\;
           {\bf next};
         }  
         \eIf{$\#(\pset{P}^{(k)}) = 1$}
         {
            $T_2 := \pop(\pset{P}^{(k)})$, $T_1 := \pop(\pset{Q}^{(k)})$\;
            $(H_2, \ldots, H_r):= \srs(\max(T_1, T_2), \min(T_1, T_2))$\;
            \For{$i=2, \ldots, r-1$}
            {
               $\Phi \!:=\! \Phi \cup \{(\pset{P}\setminus \{T_2\} \!\cup\! \{\pquo(T_2, H_i), \ini(H_{i+1}), \ldots, \ini(H_r)\}, \pset{Q}\cup \{\ini(H_i)\}, k)\}$\label{line:reg-case3-1}\;
            }
            \eIf{$x_k \in \supp(H_r)$}  
            {
               $\pset{Q}' = \pset{Q} \cup \{\ini(H_r)\}$\;
            }
            {
               $\pset{Q}' = \pset{Q} \setminus \{T_1\} \cup \{\ini(H_r)\}$\;
            }
          $\Phi := \Phi \cup \{(\pset{P} \setminus \{T_2\} \cup \{\pquo(T_2, H_r)\}, \pset{Q}', k) \}$\label{line:reg-case3-2}\;
         }
         {
            \For{$Q \in \pset{Q}^{(k)}$}
            {
                $\Phi := \Phi \cup \{(\pset{P} \cup \{\ini(Q)\}, \pset{Q}\setminus \{Q\} \cup \{\tail(Q)\}, k)\}$\label{line:reg-case4-1}\;
                $\pset{Q} := \pset{Q} \cup \{\ini(Q)\}$\label{line:reg-case4-2}\;
            }
            $\Phi := \Phi \cup \{(\pset{P}, \pset{Q}, k-1)\}$\label{line:replace-end}\;
         }

\end{algorithm}

The replacement with lines~\ref{line:replace-start}--\ref{line:replace-end} in Algorithm~\ref{alg:reg} are for, when $\#\pset{P}^{(k)} =0$ or $1$ (which means the reduction with respect to $x_k$ is done), handling the polynomials in the inequation set $\pset{Q}$. When $\#\pset{P}^{(k)} = 1$, the splitting introduced in lines~\ref{line:reg-case3-1} and \ref{line:reg-case3-2} and for a node $(\pset{P}, \pset{Q}, k)$ results in its child nodes $(\pset{P}_i, \pset{Q}_i, k)~(i=2, \ldots, r)$ in the decomposition tree, where $\pset{P}_i$ and $\pset{Q}_i$ are summarized below. 
\begin{equation}
  \label{eq:reg-def}
  \begin{split}
    \pset{P}_i &: = \left\{
      \begin{tabular}[l]{ll}
$\pset{P}\setminus \{T_2\} \cup \{\pquo(T_2, H_i), \ini(H_{i+1}), \ldots, \ini(H_r)\}$, & $i=2, \ldots, r-1$ \\
        $\pset{P}\setminus \{T_2\} \cup \{\pquo(T_2, H_r)\}$, & $i=r$,
      \end{tabular}
    \right.\\
    \pset{Q}_i  &: =  \left\{
      \begin{tabular}[l]{ll}
        $\pset{Q}\cup \{\ini(H_i)\}$, & $i=2, \ldots, r-1$, \\
        $\pset{Q} \cup \{\ini(H_r)\}$ & $i=r$ and $x_k \in \supp(H_r)$, \\
        $\pset{Q} \setminus \{T_1\} \cup \{\ini(H_r)\}$ & $i=r$ and $x_k \not \in \supp(H_r)$,
      \end{tabular}
      \right.
        \end{split}
\end{equation}
when $\#\pset{P}^{(k)} = 0$,  the splitting for $(\pset{P}, \pset{Q}, k)$ in lines~\ref{line:reg-case4-1} and \ref{line:replace-end} results in its child nodes $(\pset{P}', \pset{Q}', k)$ and $(\pset{P}_Q, \pset{Q}_Q, k)~(\forall Q\in \pset{Q})$ in the decomposition tree as follows. 
\begin{equation}
  \label{eq:reg-def-2}
  \begin{split}
    \pset{P}' & := \pset{P}, \quad \pset{Q}' := \pset{Q} \cup \{\ini(Q):\, Q\in \pset{Q}\},\\ 
    \pset{P}_{Q} &:= \pset{P} \cup \{\ini(Q)\}, \quad \pset{Q}_{Q} := \pset{Q} \setminus \{Q\} \cup \{\tail(Q)\}, \quad\forall Q \in \pset{Q}^{(k)}.\\
        \end{split}
\end{equation}

Since the new lines~\ref{line:replace-start}--\ref{line:replace-end} in Algorithm~\ref{alg:reg} can be viewed as post-processing after the reduction with respect to a variable $x_k$ is finished, the underlying difference between the decomposition tree of Algorithm~\ref{alg:reg} and that of Algorithm~\ref{alg:srs} is that the former tree has additional layers to represent the splittings described above, succeeding the bottom layer, where each node $(\pset{P}, \pset{Q}, k)$ is such that $\#\pset{P}^{(k)} < 1$, in each section for the corresponding variable $x_k = x_n, \ldots, x_1$ of the latter tree. These additional layers are illutrated in the Figure~\ref{fig:decTreeReg}. 

    \begin{figure}[ht]
      \centering
      \includegraphics[width=14cm,keepaspectratio]{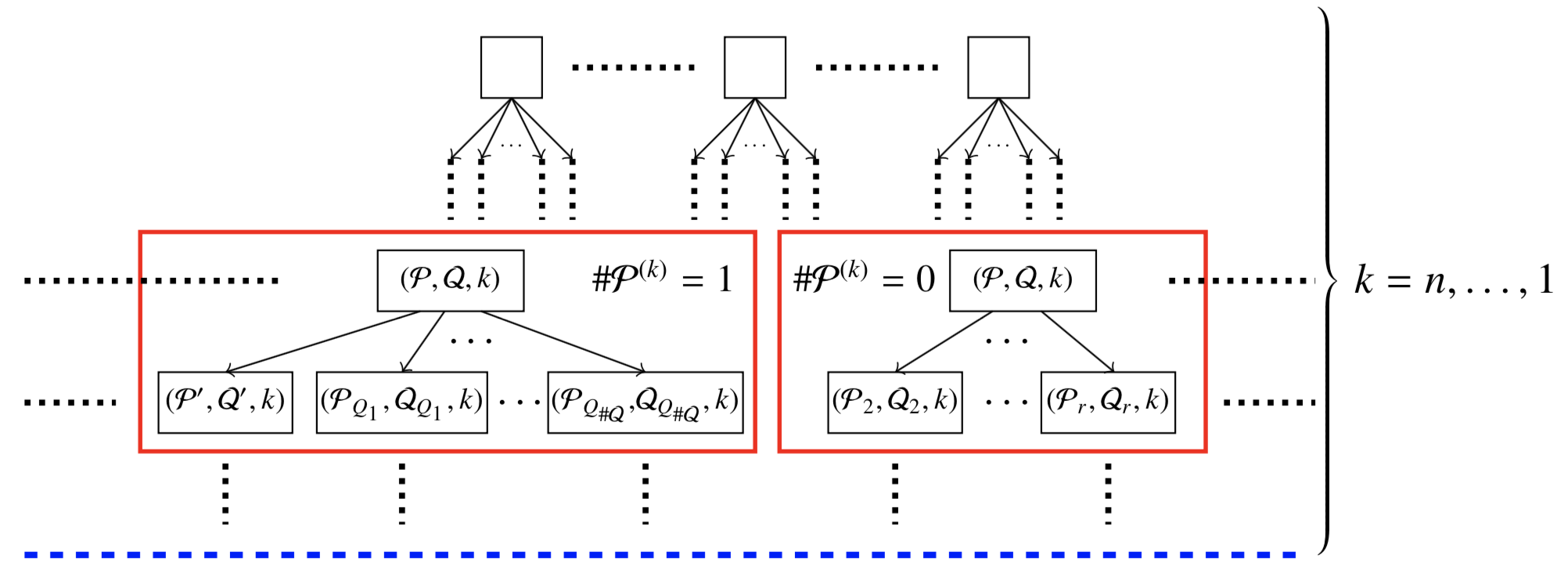}
      \caption{Additional layers in the decomposition tree of $\algReg(\pset{F})$ compared with that of $\algSub(\pset{F})$}
\label{fig:decTreeReg}
\end{figure}

\subsection{Chordality of polynomial sets in subresultant-based algorithm for regular decomposition}
\label{sec:chordalReg}

We prove the following theorem for the subresultant-based algorithm for regular decomposition as the analogue to Theorems~\ref{thm:wang} and \ref{thm:chordalSRS} for the two algorithms presented in Sections~\ref{sec:wang} and \ref{sec:subres}. 

\begin{theorem}
  \label{thm:chordalReg}
  Let $\pset{F} \subset \kx$ be a chordal polynomial set with $x_1 < \cdots < x_n$ as one perfect elimination ordering of $G(\pset{F})$. Then for any node $(\pset{P}, \pset{Q}, k)$ in the decomposition tree of $\algReg(\pset{F})$, we have $G(\pset{P}) \subset G(\pset{F})$ and $G(\pset{Q}) \subset G(\pset{F})$.
\end{theorem}

\begin{proof}
  We make induction on the depth of the node in the decomposition tree. The root of the decomposition tree is $(\pset{F}, \emptyset, n)$ and obviously the conclusion holds. Now suppose that for any node $(\pset{P}', \pset{Q}', k)$ of depth $d$ in the decomposition tree, we have $G(\pset{P}') \subset G(\pset{F})$ and $G(\pset{Q}') \subset G(\pset{F})$. Let $(\pset{P}, \pset{Q}, k)$ be of depth $d+1$ and $(\pset{P}', \pset{Q}', k')$ be its parent node of depth $d$ in the decomposition tree. Next we prove that $G(\pset{P}) \subset G(\pset{F})$ and $G(\pset{Q}) \subset G(\pset{F})$.

  Note that there are four cases where a new child node is generated from its parent node: case (1) corresponds to line~\ref{line:Pprime} of Algorithm~\ref{alg:srs}, case (2) to lines~\ref{line:sub-split} and \ref{line:end} of Algorithm~\ref{alg:srs}, case (3) to lines~\ref{line:reg-case3-1} and \ref{line:reg-case3-2} of Algorithm~\ref{alg:reg}, and case (4) to lines~\ref{line:reg-case4-1} and \ref{line:reg-case4-2} of Algorithm~\ref{alg:reg}. The conclusion has been proved for case (1) with Proposition~\ref{prop:wang-right} and for case (2) with Proposition~\ref{prop:td-srs}.

Next we prove the conclusion for case (3), where the node $(\pset{P}, \pset{Q}, k)$ is generated in line~\ref{line:reg-case3-1} or \ref{line:reg-case3-2} of Algorithm~\ref{alg:reg}. In this case, there exist polynomials $T_1' \in \pset{Q}^{(k')}$ and $T_2' \in \pset{P}^{(k')}$ such that $H_2', \ldots, H_r'$ are the subresultant regular subchain of $T_1'$ and $T_2'$. If $(\pset{P}, \pset{Q}, k)$ corresponds to the node $(\pset{P}_i, \pset{Q}_i, k)$ for some integer $i~(2 \leq i < r)$ in \eqref{eq:reg-def}, then 
  $$ \pset{P} = \pset{P'}\setminus \{T_2'\} \cup \{\pquo(T'_2, H'_i), \ini(H'_{i+1}), \ldots, \ini(H'_r)\}.$$
  Clearly $\supp(\pset{P}) \subset \supp(\pset{P}') \cup \supp(\pset{Q}') \subset \supp(\pset{F})$. For any $(x_p, x_q) \in G(\pset{P})$, if there exists a polynomial $T' \in \pset{P}' \setminus \{T_2'\}$ such that $x_p, x_q \in \supp(T')$, then $(x_p, x_q) \in G(\pset{P}') \subset G(\pset{F})$. Otherwise there exists a polynomial $T \in \{\pquo(T'_2, H'_i), \ini(H'_{i+1}), \ldots, \ini(H'_r)\}$ such that $x_p, x_q \in \supp(T) \subset \supp(T_1') \cup \supp(T_2')$. Furthermore, there are three cases: (a) If $x_p, x_q \in \supp(T_1')$, then $(x_p, x_q) \in G(\pset{Q}') \subset G(\pset{F})$ by the inductive assumption. (b) If $x_p, x_q \in \supp(T_2')$, then $(x_p, x_q) \in G(\pset{P}') \subset G(\pset{F})$ by the inductive assumption. (c) If, without loss of generality, $x_p \in \supp(T_1')$ and $x_q \in \supp(T_2')$, by $T_1' \in \pset{Q}'^{(k')}$ we know $x_p \leq x_{k'}$ and $x_p, x_{k'} \in \supp(T_1')$ and thus $(x_p, x_{k'}) \in G(\pset{Q}') \subset G(\pset{F})$ by the inductive assumption; by $T_2' \in \pset{P}'^{(k')}$ we know similarly that $x_q \leq x_{k'}$, $(x_q, x_{k'}) \in G(\pset{P}') \subset G(\pset{F})$. Then by the chordality of $\pset{F}$ we have $(x_p, x_q) \in G(\pset{F})$. Now we prove $G(\pset{P}) \subset G(\pset{F})$ for $\pset{P} = \pset{P}_i$ for some integer $i~(2 \leq i < r)$ in \eqref{eq:reg-def}. When $(\pset{P}, \pset{Q}, k)$ corresponds to $(\pset{P}_r, \pset{Q}_r, k)$ in \eqref{eq:reg-def}, the inclusion $G(\pset{P}) \subset G(\pset{F})$ can be proved in the same way as above, so can the inclusion $G(\pset{Q}) \subset G(\pset{F})$, with the observation that $\supp(\pquo(T_2', H_r')), \supp(\ini(H_i')) \subset \supp(T_1') \cup \supp(T_2')$ for $i=2, \ldots, r$. 

The conclusion can be derived for case (4), where the node $(\pset{P}, \pset{Q}, k)$ is generated in line~\ref{line:reg-case4-1} or \ref{line:reg-case4-2} of Algorithm~\ref{alg:reg}, with easy verification with \eqref{eq:reg-def-2}. 

\end{proof}

  \begin{corollary}\label{cor:srs-ts}
    Let $\pset{F} \subset \kx$ be a chordal polynomial set with $x_1 < \cdots < x_n$ as one perfect elimination ordering of $G(\pset{F})$ and $\pset{T}_1, \ldots, \pset{T}_r$ be the triangular sets computed by $\algReg(\pset{F})$. Then $G(\pset{T}_i) \subset G(\pset{F})$ for $i=1, \ldots, r$. 
  \end{corollary}

  \begin{proof}
    Straightforward from Theorem~\ref{thm:chordalReg}. 
  \end{proof}

  \begin{remark}\rm
    Simple sets, also called squarefree regular sets, are a stronger kind of triangular sets than regular ones (there are corresponding concepts of simple systems too) \cite{w98d,LI2010D,M2012d}. They are particularly uesful for counting the numbers of solutions of polynomial systems and have been applied to study differential equations \cite{B14c,B2011a,GR2016l}. An algorithm in top-down style is proposed in \cite{w98d} for decomposing an arbitrary polynomial set into simple systems over fields of characteristic zero by using subresultant regular subchains. This algorithm is similar to Algorithm~\ref{alg:reg} and the authors have the feeling that it should also ``preserve'' chordality of the input polynomial set. But it seems not easy to formulate the decomposition process and the corresponding decomposition tree of this algorithm as ``cleanly'' as what have been done here in this paper for Algorithm~\ref{alg:reg}. 
  \end{remark}


\section{Sparse triangular decomposition in top-down style}
\label{sec:sparse}

\subsection{Variable sparsity of polynomial sets}
\label{sec:sparsity}

When speaking of a sparse polynomial set $\pset{F} \subset \kx$, one usually means that the percentage of terms effectively appearing in $\pset{F}$ in all the possible terms in the variables $x_1, \ldots, x_n$ up to a certain degree is low. This kind of sparsity for polynomial sets is convenient for the computation of \grobner bases which is essentially based on reduction with respect to terms. In fact, efficient algorithms for computing \grobner bases for sparse polynomial sets defined in this way have been proposed, implemented, and analyzed \cite{FSS2014s,BFT18t}. 

Instead of terms of polynomials, triangular sets focus on the variables of polynomials. As exploited in \cite{C2017c}, the sparsity of the polynomial sets with respect to their variables is reflected in their associated graphs. 

\begin{definition}\rm
  \label{def:var-sparsity}
Let $\pset{F} \subset \kx$ be a polynomial set and $G(\pset{F}) = (V, E)$ be its associated graph. Then the \emph{variable sparsity} $s_v(\pset{F})$ of $\pset{F}$ is defined to be
\begin{equation}
  \label{eq:sparsity}
s_v(\pset{F}) = |E| / \binom{2}{|V|},  
\end{equation}
where the denominator is the number of edges of a complete graph composed of $|V|$ vertices.
\end{definition}

A similar notion called the correlative sparsity was proposed in \cite{WKKM06s} to represent the variable sparsity of a multivariate polynomial for the purpose of decomposing it into sums of squares. 

The associated graph $G(\pset{F})$ can be extended to a weighted one $G^w(\pset{F})$ by associating the number $\#\{F\in \pset{F}: x_i, x_j \in \supp(F)\}$ to each edge $(x_i, x_j)$ of $G(\pset{F})$. Let $G^{w}(\pset{F}) = (V, E)$, with the weight $w_e$ for each $e\in E$, be the weighted associated graph of $\pset{F}$. Then the \emph{weighted variable sparsity} $s_v^w(\pset{F})$ of $\pset{F}$ can be defined as

\begin{equation}
  \label{eq:WeightSparsity}
s_v^w(\pset{F}) = \frac{\sum_{e\in E} w_e}{\#\pset{F}\cdot \binom{2}{|V|}}.
\end{equation}

\begin{example}\rm
For the following polynomial set  \begin{equation*}
\begin{aligned}
\pset{F}=\{&x_1x_4-x_2x_3,x_2x_5-x_3x_4,x_3x_6-x_4x_5,x_4x_7-x_5x_6,x_5x_8-x_6x_7,x_6x_9-x_7x_8,\\
&x_7x_{10}-x_8x_9,x_8 x_{11}-x_9x_{10},x_9x_{12}-x_{10}x_{11},x_{10}x_{13}-x_{11}x_{12},x_{11}x_{14}-x_{12}x_{13},\\
&x_{12}x_{15}-x_{13}x_{14},x_{13}x_{16}-x_{14}x_{15},x_{14}x_{17}-x_{15}x_{16},x_{15}x_{18}-x_{16}x_{17},x_{16}x_{19}-x_{17}x_{18}\},
\end{aligned}
\end{equation*}
its variable sparsity is $s_v(\pset{F})=|51|/\binom{2}{|19|} \approx 0.298$ by definition. 
\end{example}

\begin{example}\rm
  Consider the polynomial set
  $$ \pset{F} = \{x_1\cdots x_6-1, x_1^2+x_2, x_2^2+x_3, x_3^2+4, x_4^2+x_5, x_5^2+x_6\}.$$
  Then one can easily check that its variable sparsity $s_v(\pset{F}) = 1$ but the weighted variable sparsity $s_v^w(\pset{F}) = 20/\binom{2}{|6|} \approx 0.22$.
  This example effectively demonstrates the difference between the notions of the variable sparsity and the weighted one. 
\end{example}

\subsection{Algorithm for sparse triangular decomposition in top-down style}
\label{sec:sparseTD}

With the introduction of variable sparsity of polynomial sets above, what we have proved in Theorems \ref{thm:wang}, \ref{thm:chordalSRS}, and \ref{thm:chordalReg} for the three studied algorithms for triangular decomposition in top-down style can be further viewed in the following way. Given a polynomial set which is sparse with respect to the variables, if it is chordal and we use a perfect elimination ordering of its associated graph as the variable ordering, then all the polynomial sets in the process of triangular decomposition in top-down style with these algorithms keep to be sparse with respect to the variables. In a word, variable sparsity is preserved in these algorithms for triangular decomposition when the input polynomial set is chordal and sparse with respect to the variables.

Inspired by the utilization of chordal structures in Gaussian elimination of sparse matrices, the algorithm below for sparse triangular decomposition in top-down style is proposed. We use regular decomposition for example, and in fact one can use any of the three algorithms studied in this paper for the corresponding sparse triangular decomposition.

In Algorithm~\ref{alg:sparse} below, we modify the algorithm $\algReg()$ a little so it now takes two arguments, with the second as the variable ordering used for computing the regular decomposition of the first argument. The constant $s_0$ is a threshold for determining whether a polynomial set is sparse with respect to the variables, and one may take $s_0 = 0.3$ for example. As mentioned in the Section~\ref{sec:graph}, testing whether a graph is chordal and finding a chordal completion can be realized effectively, for example with the algorithms described in \cite{RTL76a} and \cite{BK08c} respectively.

\begin{algorithm}[!ht]
  \caption{Algorithm for sparse regular decomposition in top-down style
    $~~~~~~~~~~~~~~~~~~~~~~\Psi :=\algSparse(\pset{F})$}
\label{alg:sparse}

\small

\KwIn{$\pset{F}$, a polynomial set in $\kx$}

\KwOut{$\overline{\p{x}}$, a variable ordering of $x_1, \ldots, x_n$; $\Psi$, a set of finitely many regular systems which form a regular decomposition of $\pset{F}$ with respect to $\overline{\p{x}}$}

\BlankLine

$s_v(\pset{F}) :=$ variable sparsity of $\pset{F}$\;

\eIf{$s_v(\pset{F}) < s_0$}
{
  \eIf{$G(\pset{F})$ is chordal}
  {
    $\overline{\p{x}} :=$ perfect elimination ordering of $G(\pset{F})$\;
  }
  {
    $G' :=$ a chordal completion of $G(\pset{F})$\;
    $\overline{\p{x}} :=$ perfect elimination ordering of $G'$\;
  }  
}
{
  $\overline{\p{x}} :=$ a random ordering of $x_1, \ldots, x_n$\;
}
$\Psi := \algReg(\pset{F}, \overline{\p{x}})$\;

\Return $(\Psi, \overline{\p{x}})$\;
      
\end{algorithm}

\subsection{Preliminary experimental results}
\label{sec:exp}

In this part we report some preliminary experimental results with the proposed algorithm for sparse triangular decomposition in top-down style (Algorithm~\ref{alg:sparse}) applied to two sparse and chordal polynomial sets in \cite{C2017c}. Note that these two polynomial systems are very special, and we mainly want to demonstrate the increasing effectiveness of Algorithm~\ref{alg:sparse} when the input polynomial sets become sparser with respect to the variables and to provide theoretical explanations for the experimental observation in \cite{C2017c}. More experiments with the algorithm on sparse and chordal polynomial sets are our future work.

We apply Algorithm~\ref{alg:sparse} to polynomial sets of different sizes of these two kinds to compute the regular decomposition of them. Here the function {\sf RegSer} in the software pacakge {\sc Epsilon} (version 0.618) for the computer algebra system {\sc Maple} is called as the implementation of the algorithm for regular decomposition in top-down style (Algorithm~\ref{alg:reg}). We also try random variable orderings for regular decomposition of the same polynomial sets and compare the timings. All these experiments are carried out on a laptop with the following resources: 1.6GHz Intel Core i5 CPU, 4GB 1600MHz DDR3 memory, macOS Sierra 10.12.5 operating system, and {\sc Maple} version 2017.0.

\subsubsection{Lattice reachability problem}
\label{sec:exp-lattice}

The following polynomial set from a lattice reachability problem is described in \cite{DES98l}. 
\begin{equation} \label{lizi1}
  \pset{F}_i:=\{x_{k}x_{k+3}-x_{k+1}x_{k+2}: k=1,2,\ldots,i \ \},\quad i\in \mathbb{Z}.
\end{equation}
The associated graph of $\pset{F}_i$ is shown in Figure~\ref{fig:lattice}~(left). One can find from this figure that $G(\pset{F}_i)$ is indeed chordal. It is easy to verify that the variable sparsity for $\pset{F}_i$ is 
\begin{equation*} 
s_v=\frac{6n-12}{n^2-n},
\end{equation*}
where $n$ is the number of variables in $\pset{F}_i$. This expression shows that the variable sparsity decreases as the number of variables grows. In particular, when $n \geq 17$, we have $s_v(\pset{F}_i) \leq 0.3$.

We pick $i=5, 7, 17, 22, 27, 32, 37$ in $\pset{F}_i$ in \eqref{lizi1} to form polynomial sets of different sizes and then apply Algorithm~\ref{alg:sparse} to compute their regular decomposition. The timings of our experiments are recorded in the following table (Table~\ref{tab:lattice}), where the column labeled with $n$ records the number of variables, that with $s_v$ is for the variable sparsity, that with $t_p$ is for timings of regular decomposition with respect to a perfect elimination ordering, that with $t_r$ is for timings with respect to random variable orderings, that with $\overline{t}_r$ is for the average timings of five random orderings, and that with $\overline{t}_r/t_p$ is for the speedup ratios.

\begin{figure}
  \label{tab:lattice}
      \centering
\includegraphics[width=1.6cm,keepaspectratio]{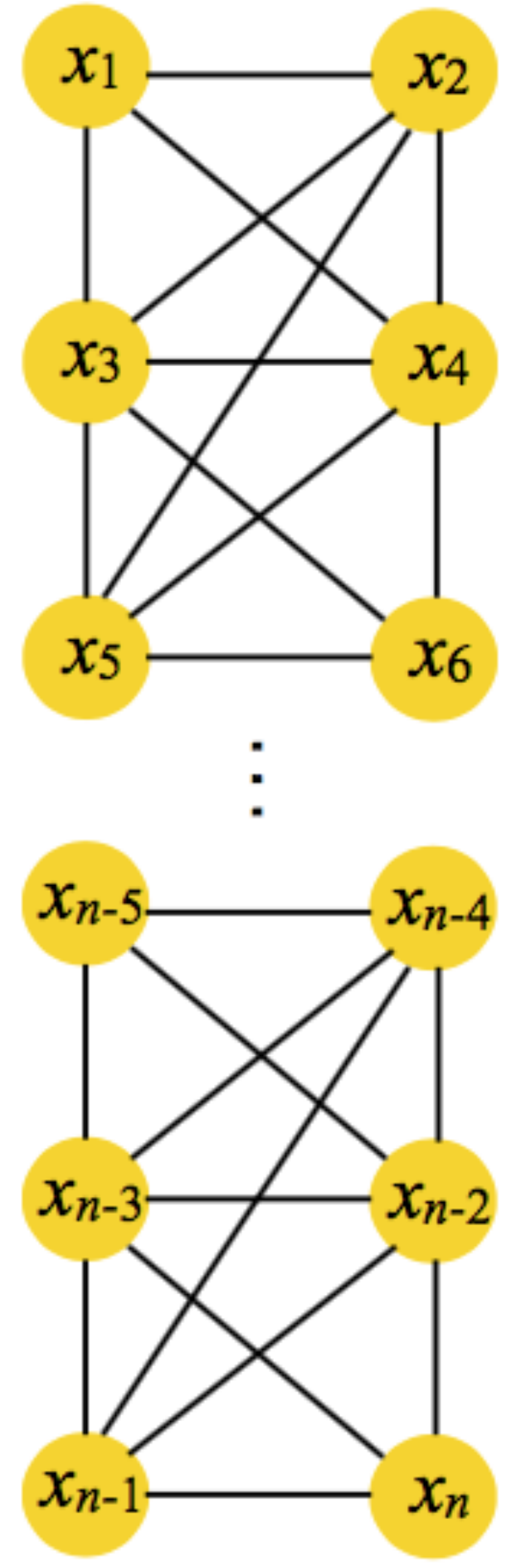}~\qquad\qquad\qquad \includegraphics[width=1.6cm,keepaspectratio]{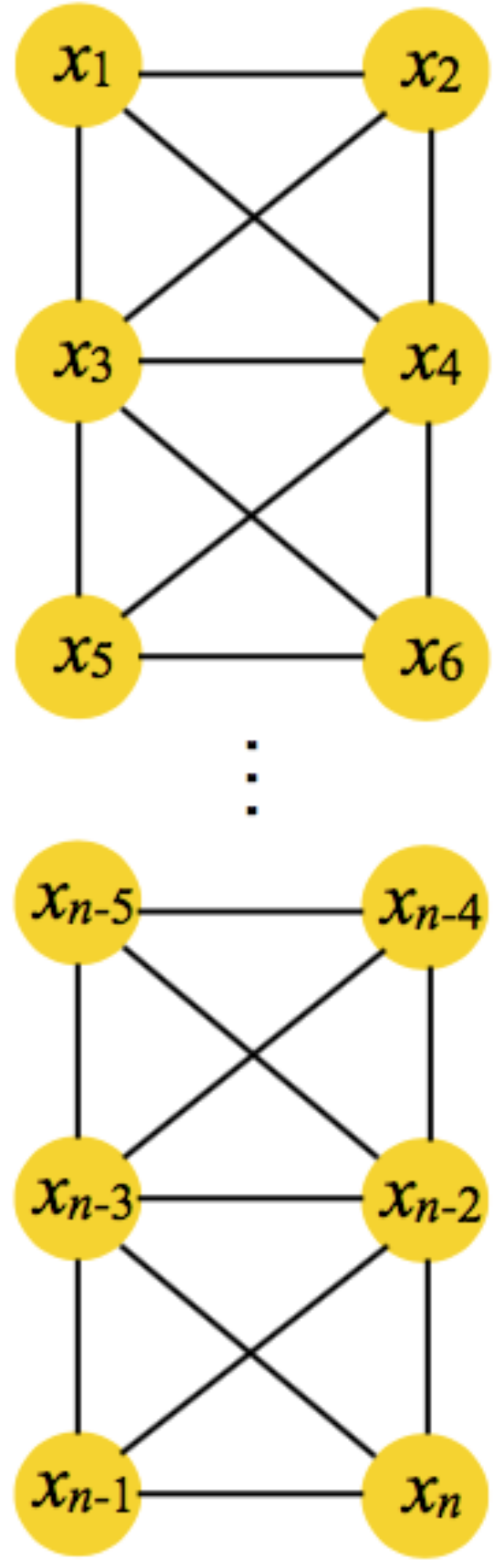}~
      \caption{Associated graphs $G(\pset{F}_i)$ (left) for $\pset{F}_i$ in \eqref{lizi1} and $G(\tilde{\pset{F}_i})$ (right) for $\tilde{\pset{F}_i}$ in \eqref{lizi2}}
      \label{fig:lattice}
    \end{figure}

    \begin{table}[!h]
      \caption{Timings (in seconds) for regular decomposition of $\pset{F}_i$ with different variable orderings (in top-down style)}
            \label{tab:lattice}

\begin{tabular}{cccccccccc}
\hline
$n$&$s_v$&$t_p$&\multicolumn{4}{c}{$\qquad\qquad\quad t_r$}&&$\overline{t}_r$&$\overline{t}_r/t_p$\\
\hline
8&0.64&0.11&0.10&0.09&0.05&0.06&0.09&0.10&0.91\\

10&0.53&0.19&0.14&0.21&0.22&0.11&0.21&0.18&0.95\\

20&0.28&1.44&4.24&4.45&3.15&4.41&4.65&4.18&2.90\\

25&0.23&4.25&50.62&20.29&15.55&25.01&35.10&29.31&6.90\\

30&0.19&11.94&177.37&185.94&130.54&142.97&103.42&148.05&12.40\\

35&0.17&42.33&560.56&291.35&633.43&320.98&938.45&548.95&12.97\\

40&0.15&161.11&1883.64&3618.04&4289.13&4013.99&2996.37&3360.23&20.86\\
\hline
\end{tabular}
\end{table}

As can be seen from Table~\ref{tab:lattice}, when the polynomial system is not sparse (say, for $n=8$ or $10$), regular decomposition in top-down style with respect to a perfect elimination ordering does not bring efficiency gains versus a random variable ordering. However, when the polynomial systems become sparser increasingly (see column $s_v$), the efficiency gains with choosing a perfect elimination ordering become more and more prominent. This means that, for this particular kind of polynomial sets, the variable sparsity is utilized by picking a good variable ordering (the perfect elimination ordering) in algorithms for triangular decomposition in top-down style to make them more efficient.

We also tried regular decomposition of some of these polynomial sets above with algorithms not in top-down style. Here the function {\sf Triangularize} in the built-in package {\sc RegularChains} in {\sc Maple} is called as an implementation of the algorithm for regular decomposition not in top-down style. The timings with these experiments are recorded in Table~\ref{tab:notTopDown}. In this table one can find that a perfect elimination ordering will not bring speedups against a random variable ordering for this particular kind of polynomial sets.

\begin{table}[!h]
  \caption{Timings (in seconds) for regular decomposition of $\pset{F}_i$ with different variable orderings (NOT in top-down style)}
  \label{tab:notTopDown}

\begin{tabular}{cccccccccc}
\hline
$n$&$s_v$&$t_p$&\multicolumn{4}{c}{$\qquad\qquad\quad t_r$}&&$\overline{t}_r$&$\overline{t}_r/t_p$\\
\hline
15&0.37&45.90&17.29&21.41&13.62&32.50&19.63&20.89&0.46\\

17&0.33&216.69&87.29&197.35&104.86&68.28&130.83&117.72&0.54\\

19&0.30&1303.08&415.90&308.37&780.75&221.75&831.15&511.58&0.39\\

21&0.27&8787.32&1823.29&2064.55&2431.49&1926.02&1593.36&1967.74&0.22\\
\hline
\end{tabular}
\end{table}

\subsubsection{Adjacent Minors}
\label{sec:exp-adjacent}

For the following polynomial set arising from problems related to adjacent minors of matrices
\begin{equation} \label{lizi2}
\tilde{\pset{F}}_i:=\{x_{2k-1}x_{2k+2}-x_{2k}x_{2k+1}: k=1,2,\ldots,i \},\quad i\in \mathbb{Z},
\end{equation}
one can find that its variable sparsity is 
\begin{equation*} \label{2xishudu}
s_v=\frac{5n-8}{n^2-n},
\end{equation*}
where $n$ is the number of variables in the polynomial set. This polynomial set and the expression of its variable sparsity are similar to those in Section~\ref{sec:exp-lattice}. We pick $i=3, 7, 9, 11, 13, 15$ in $\tilde{\pset{F}_i}$ in \eqref{lizi2} to form polynomial sets of different sizes for regular decomposition and the experimental results are reported in Table~\ref{tab:adjacent}. 

\begin{table}[!h]
    \caption{Timings (in seconds) for regular decomposition of $\tilde{\pset{F}}_i$ with different variable orderings (in top-down style)}
   \label{tab:adjacent}
\begin{tabular}{cccccccccc}
\hline
$n$&$s_v$&$t_p$&\multicolumn{4}{c}{$\qquad\qquad\quad t_r$}&&$\overline{t}_r$&$\overline{t}_r/t_p$\\
\hline
8&0.57&0.10&0.09&0.03&0.04&0.05&0.03&0.05&0.50\\

16&0.3&0.64&2.20&1.29&1.77&1.54&2.08&1.78&2.78\\

20&0.24&2.04&9.47&9.38&22.05&14.06&12.31&13.45&6.59\\

24&0.20&9.05&67.00&39.37&109.12&59.21&70.67&69.07&7.63\\

28&0.17&7.61&63.02&66.63&66.78&110.17&74.23&76.17&10.01\\

32&0.15&133.45&1173.39&3371.60&1457.92&1095.27&1469.47&1713.53&12.84\\
\hline
\end{tabular}
\end{table}

\begin{remark}\rm
  \label{rem:explain}
  At this point, we are able to theoretically explain the experimental observation in \cite{C2017c} that algorithms due to Wang become more efficient when the polynomial sets to decompose are chordal. Most of the algorithms for triangular decomposition due to Wang are in top-down style, and in this paper we have proven that a few algorithms for triangular decomposition in top-down style ``preserve'' chordality of the input polynomial sets. With the notion of variable sparsity of polynomial sets, this theoretical property means that, if the polynomial set is chordal and sparse with respect to the variables, the variable sparsity will be preserved throughout the process of triangular decomposition. The computation of triangular decomposition carried out in \cite{C2017c} with Wang's algorithms are mainly on chordal and sparse (with respect to the variables) polynomial sets and the perfect elimination orderings are used as the variable orderings (which is also important, as shown in Tables~\ref{tab:lattice} and \ref{tab:adjacent}), this scenario makes Wang's algorithms the so-called ``algorithms for sparse triangular decomposition'' in this paper, which make full use of the variable sparsity of polynomial sets and have remarkable efficiency gains in case of chordal and very sparse polynomial sets with respect to the variables. 
\end{remark}

{\small
\bibliographystyle{abbrv}
\bibliography{Chordal-extended}
}

\end{document}